\RequirePackage{fix-cm}
\documentclass[smallextended,natbib]{svjour3}       
\smartqed  

\usepackage{amsmath}
\usepackage{amssymb}
\usepackage[utf8]{inputenc}
\usepackage{color}
\usepackage{algorithm}
\usepackage{algpseudocode}
\usepackage{tikz}
\usepackage{mathtools}
\usepackage{multirow}
\usepackage{comment}
\usepackage{booktabs}
\usepackage{soul}
\usepackage{graphicx}
\usepackage[dvipsnames]{xcolor}

\usepackage{newtxtext,newtxmath}
\journalname{Journal of Combinatorial Optimization}

\newtheorem{observation}[theorem]{Observation}

\allowdisplaybreaks[3]

\begin{document}

\title{Complexity and Algorithms for Unary Translocation Distance \thanks{A preliminary version of this article was accepted at International Symposium on Combinatorial Optimization (ISCO) 2026.}}


\author{Maria Constantin \and Adrian Micl\u{a}u\c{s} \and Alexandru Popa \and Andrei Popa}

\institute{
Maria Constantin \and Adrian Micl\u{a}u\c{s} \and Alexandru Popa \and Andrei Popa
\at Department of Computer Science, University of Bucharest, Str.~Academiei 14, Bucharest, 010014, Romania\\
\email{maria.petruta.constantin@drd.unibuc.ro, adrian.miclaus@fmi.unibuc.ro, alexandru.popa@fmi.unibuc.ro, andreipopa0498@gmail.com}
}

\date{Received: date / Accepted: date}

\maketitle

\begin{abstract}

Given a finite set of integers $A$, a \emph{unary translocation} produces a new set
$A' = A \cup \{u,v\}$, where $u$ and $v$ are nonnegative integers satisfying
$x+y=u+v$ for some $x,y\in A$. For an input set $A$ and a target set $B$, the
\emph{unary translocation distance} is the minimum number of unary translocations
required to obtain a superset containing $B$. In this paper, we study this problem
from both theoretical and computational perspectives. We prove that computing the unary translocation distance is strongly NP-hard,
thereby answering an open question raised by \citet{ConstantinMiclausPopa2026UnaryTranslocation}. On the positive side, we
give an exact pseudo-polynomial algorithm for every fixed constant value of
$|B|$, extending our previous results for $|B|\leq 2$. For arbitrary target
sets, we present a $2$-approximation algorithm, an additive
$(|B|-1)$-approximation algorithm, and show that the additive algorithm also
yields a $3$-approximation. We also propose parameterized algorithms, including
algorithms parameterized by the maximum value in the input set together with the
optimum distance, and by the maximum value in the target set together with
$|B|$. In addition, we propose an integer linear programming formulation that
gives an exact mathematical model for the problem, analyze its size, and show
that the LP relaxation has integrality gap at least $\frac{4}{3}$.
Finally, we report computational experiments comparing the $2$-approximation
algorithm, beam search, and simulated annealing. The results show that the
approximation algorithm is highly effective in practice and often outperforms the
heuristic baselines.

\end{abstract}

\section{Introduction}

In this paper we study the following problem, which we call the \emph{unary translocation distance problem}. An instance consists of two finite sets of nonnegative integers $A,B \subseteq \mathbb{N}$. Starting from $A_0=A$, we iteratively construct sets $A_1,A_2,\dots$ by applying \emph{unary translocation operations}. More precisely, given two numbers $x,y \in A_{i-1}$, one may produce two numbers $u,v \in \mathbb{N}$ such that $x+y=u+v$, and define
$A_i = A_{i-1} \cup \{u,v\}$.
The objective is to determine the minimum number of steps $i$ for which $B \subseteq A_i$.

This numerical problem is the unary specialization of the translocation distance on strings, where the alphabet consists of a single symbol; see, e.g.,~\citep{Martin-VideM04,constantin2019some,constantin2024simulated,constantin2025exact}. In the corresponding string formulation, one starts from a set of strings $A_0=A$ and repeatedly applies \emph{translocations}: given two strings $x,y \in A_{i-1}$, a translocation swaps two prefixes of $x$ and $y$, producing two new strings $u,v$, and sets $A_i=A_{i-1}\cup\{u,v\}$. The goal is again to reach a set containing a target set $B$.

Translocations have a long history in genome rearrangement theory, where genomes are commonly modeled as permutations or signed permutations of genes and the aim is to transform one genome into another by a minimum number of rearrangement operations~\citep{kececioglu1995mice,zimaoli,hannenhalli1996polynomial,bergeron2006sorting,li2004linear,ZHU2006322,cui20081,jiang20141,pu20201,sun2025randomized}. In the classical setting, a translocation exchanges two contiguous segments belonging to different chromosomes. This area has been studied from many algorithmic perspectives: the signed case admits polynomial-time algorithms~\citep{hannenhalli1996polynomial,WangZLM05,li2004linear,bergeron2006sorting}, whereas several unsigned variants are NP-hard and have been studied extensively through approximation algorithms~\citep{kececioglu1995mice,ZHU2006322,cui20081,jiang20141,pu20201,sun2025randomized}. Translocations have also been investigated in combination with other genome rearrangement operations, such as inversions, fusions, and fissions~\citep{kececioglu1995mice,hannenhalli1995transforming,yancopoulos2005efficient}, and related models continue to attract attention~\citep{lai2025symmetric,li2024flanked}.

A key distinction between the above permutation-based literature and our setting is the treatment of duplicates. Classical genome rearrangement models typically assume that each gene appears exactly once, that is, duplicates are excluded~\citep{zeira2019genome}. Motivated by biological situations in which duplicates cannot be ignored,~\cite{Martin-VideM04} initiated the study of translocation distance in a string-based framework, where genomes are represented as finite sets of strings and a translocation swaps prefixes of two strings. In that framework, they introduced an exact algorithm for the \emph{uniform} translocation distance when the target set has size one, together with a $2$-approximation algorithm for arbitrary target sets. This direction was later extended by \cite{constantin2019some,constantin2025exact}, who gave additional exact and approximation results for contiguous variants, and by subsequent experimental work on heuristic approaches~\citep{constantin2024simulated}.

Our unary problem keeps the combinatorial structure of the string model, but focuses only on its numerical part. In particular, unlike many standard string distances such as Hamming distance or edit distance, which become essentially trivial over a unary alphabet, translocations remain nontrivial in this restriction. This makes the unary setting an interesting object of study in its own right and not merely a degenerate special case.

From a combinatorial and optimization perspective, the unary model is also closely related to addition chains and their generalizations. Classical addition chains are a central tool in fast exponentiation, while more general frameworks such as addition sequences, vectorial addition chains, and addition-subtraction chains arise naturally in multi-exponentiation and related optimization problems~\citep{gordon1998survey,downey1981computing,olivos1981vectorial,bos1989addition,de1994efficient,morain1990speeding,cohen2005handbook}. In particular, the \emph{addition sequence problem} asks for a shortest computation that simultaneously generates a prescribed set of target integers and is NP-complete~\citep{downey1981computing}. Vectorial addition chains generalize the single-target setting to multidimensional targets~\citep{olivos1981vectorial,de1994efficient}, while addition-subtraction chains enrich the available operations in order to shorten computations~\citep{morain1990speeding,gordon1998survey}. Unary translocations fit the same broad paradigm of generating new values from previously available ones by repeated binary operations. However, they also differ in a fundamental way: each operation acts on a pair of available integers and preserves their sum, and the objective is to generate an entire target set under this conservation constraint. In this sense, the unary translocation distance problem is closer to multi-target addition-sequence problems than to the classical single-target addition-chain setting~\citep{downey1981computing,olivos1981vectorial,de1994efficient}.

\section{Preliminaries}\label{sec:preliminaries_unary}

In this section we introduce the unary translocation operation on integers and the associated distance problem that we study in the rest of the paper. For a set $A \subset \mathbb{N}$, we denote by $\max A = \displaystyle  \max_{y\in A} \{y\} $.

\begin{definition}[Unary translocation]
Let $x,y, i,j \in \mathbb{N}$ be such that
$0 \le i \le x$ and $0 \le j \le y$. A unary translocation with parameters
$(i,j)$ transforms the pair $(x,y)$ into the pair $(u,v)$ given by $ u = x - i + j, v = y - j + i$
and we write $(x,y) \vdash_{i,j} (u,v)$. When it is clear from the context, or when $i$ and $j$ are not relevant, we simply write $\vdash$.
\end{definition}

This operation preserves the sum of the two numbers: $u + v = (x - i + j) + (y - j + i) = x + y$.  We define next the notion of \emph{translocation sequence}.

\begin{definition}[Translocation sequence]
A \emph{translocation sequence over $A$} is a finite sequence
$ S = \{ (x_1,y_1) \vdash (u_1,v_1), \dots, (x_n,y_n) \vdash (u_n,v_n) \}$
such that $x_i, y_i \in A_{i-1}(S)$, $\forall i$, $1 \leq i \leq n$, where sets $A_0(S), \dots, A_n(S)$ are defined as follows: $A_0(S) = A$ and $A_i(S) = A_{i - 1}(S) \cup \{u_i, v_i\}$, $\forall i$, $1 \leq i \leq n$.  When $S$ is clear from the context, we simply write $A_0, A_1, \dots, A_n$ instead of $A_0(S), A_1(S), \dots, A_n(S)$.
\end{definition}

Now we are ready to define the main problem studied in this paper.

\begin{problem}[Unary translocation distance]
Given two disjoint finite sets of nonnegative integers $A,B \subseteq \mathbb{N}$, $A\cap B = \emptyset$,
the unary translocation distance from $A$ to $B$ is $ \mathit{TD}(A,B) = \min \{|S| \mid S \text{ is a translocation sequence over } A \text{ and } B \subseteq A_{|S|}(S) \}$.
The unary translocation distance problem asks to compute $\mathit{TD}(A,B)$.
\end{problem}

In the literature on translocations there are several modelling choices
(contiguous vs.\ non-contiguous, uniform vs.\ non-uniform). In the
\emph{non-contiguous} model, produced strings can be reused, while
\emph{non-uniform} allows arbitrary lengths. In this work we adopt the
\emph{non-contiguous, non-uniform} unary model. For formal definitions in
the general string setting we refer to~\cite{constantin2025exact}.

\section{NP-hardness}


Let $G = (V, E)$ be an undirected graph without self loops, and $N = |V|$. We assume that $V = \{1, 2, \dots, N\}$ and for every $(i, j) \in E$ there is also $(j, i) \in E$. We reduce the problem of finding a Hamiltonian path from node $1$ to node $N$ to finding the unary translocation distance between two sets $A$ and $B$.

A sequence $(a_i)_{i=1}^n$ is called a $B_3$ sequence if $a_i + a_j + a_k$ are all different for all $1 \leq i \leq j \leq k \leq n$. Note that if $a$ is a $B_3$ sequence, by fixing $k = n$, we also have that $a$ is a $B_2$ sequence (also called a Sidon sequence), which means that $a_i + a_j$ are all different for any $1 \leq i \leq j \leq n$. Also, any subsequence of a $B_3$ sequence and any reordering of a $B_3$ sequence are $B_3$ sequences too. It is proved in~\citep{bose1960theorems, o2012complete, NATHANSON2022133, sergeev2024additive} that for a given $n$ there are $B_3$ sequences with numbers up to $O(n ^ 3)$.

Let $(b_i)_{i=1}^{2N}$ be a $B_3$ sequence of length $2N$ such that $b_i \leq b_{i + 1}$ for all $ 1 \leq i < 2N$. For each vertex $i \in V$ let $f_i = b_i$ and $r_i = M + b_{N + i}$, where $M = 4 \cdot (b_{2N} + 1)$. Let $A = \{e_{i, j} = f_j + r_i - f_i\ \; | \; (i, j) \in E\} \cup \{f_1\}$ and $B = \{f_i, r_i \; | \; 1 \leq i \leq N\} - \{f_1, r_N\}$.

We will show that the existence of a Hamiltonian path from $1$ to $N$ in $G$ is equivalent to $\mathit{TD}(A, B) = N - 1$.

\begin{remark}
    Since $|B - A| = 2(N - 1)$, we have $\mathit{TD}(A, B) \geq N - 1$, since at every translocation can introduce at most two new numbers, so $\mathit{TD}(A, B) = N - 1$ is equivalent to the existence of a translocation sequence $S$ of size $N - 1$.
\end{remark}
First we show the existence of a Hamiltonian cycle implies $\mathit{TD}(A, B) = N - 1$.

\begin{lemma}
    \label{lem:path-implies-td}
    If there is a Hamiltonian path $a_1, a_2,\ldots,a_N$ with $a_1 = 1$ and $a_N = N$ in $G$, then $\mathit{TD}(A, B) = N - 1$.
\end{lemma}
\begin{proof}
Let $a_1, a_2, \ldots, a_N$ with $a_1 = 1$, $a_N = N$, a path in $G$ which contains all the nodes in $V$. We prove the following: there is a translocation sequence $S$ such that for any $1 \leq i \leq N$ we have $f_{a_j} \in A_{i - 1}(S)$ for all $j \leq i$ and $r_{a_j} \in A_{i - 1}(S)$ for all $j < i$.
    For $i = 1$ we already have $v_1$ in $A_0(S) = A$. To have the property hold for $i \geq 2$, we choose the $(i-1)$-th translocation in $S$ to be $(f_{a_{i - 1}}, e_{a_{i - 1}, a_i}) \vdash (f_{a_i}, r_{a_{i - 1}})$, in order for $f_{a_i}$ and $r_{a_{i - 1}}$ to be in $A_{i - 1}(S)$. The translocation is valid, since $f_{a_{i - 1}} \in A_{i - 2}(S)$, $e_{a_{i - 1},a_{i}} \in A$ because there is an edge $(a_{i - 1}, a_i) \in E$, and $f_{a_{i - 1}} + e_{a_{i - 1}, a_{i}} = f_{a_i} + r_{a_{i - 1}}$.
    \qed
\end{proof}

To prove the converse we show some intermediary results.
\begin{remark}
    In the following section we will consider a translocation $(x, y) \vdash (u, v)$ to be equivalent to $(y, x) \vdash (u, v)$, $(x, y) \vdash (v, u)$, and, implicitly, $(y, x) \vdash (v, u)$. Let $s_i = (x_i, y_i) \vdash (u_i, v_i)$ be a translocation of $S$. We prove that $s_i = (f_k, e_{k, j}) \vdash (f_j,r_k)$ for all $i \in \{ 1, 2, \ldots |S|\}$. 
\end{remark}

We begin by proving that all the values in sets $A$ and $B$ are different.

\begin{proposition}
    \label{prop:all-distinct}
    All values $f_i, r_i, e_{i,j}$ are different for $1 \leq i, j \leq N$.
\end{proposition}
\begin{proof}
    Values $f_i$ are distinct among them because they form a $B_3$ sequence. Values $r_i$ are distinct because the sequence $(r_i - M)_{i=1}^N$ is a $B_3$ sequence (and, implicitly, the sequence $(r_i)_{i=1}^N$ itself). We have $f_i < M \leq r_j$ for all $1 \leq i, j \leq N$.
    We prove that $e_{k, l} \neq r_i$ for $1 \leq i, k, l \leq N$. We have $|f_k - f_l| < b_{2N} + 1$, so $e_{k, l} = r_k + f_l - f_k \geq 3(b_{2N} + 1)$, thus $e_{k, l} > f_i$.

    Assume $e_{k, l} = r_i$ for some $1 \leq i, k, l \leq N$.
    Then $r_k + f_l - f_k = r_i$ by adding $f_k$ on both sides and replacing with sequence $b$ we have $M + b_{N+k} + b_l = M + b_{N + i} + b_k$. Because $b$ is a $B_3$ sequence, we have $k = i$ and $l = k$, so there is an edge $e_{i, i}$, which is a contradiction.

    Assume $e_{i, j} = e_{k, l}$ for some $1 \leq i, j, k, l \leq N$. Then $r_i + f_j - f_i = r_k + f_l - f_k$. By adding $f_i + f_k$ on both sides and replacing with $b$ we obtain $M + b_{N + i} + b_j + b_k = M + b_{N + k} + b_l + b_i$. Because $b$ is a $B_3$ sequence, we have $i = k$ and either $j = l$, $k = i$ or $j = i$, $k = l$. The latter case, $i = j$, $k = l$ implies $e_{i, j}$ and $e_{k, j}$ are self loops, which is a contradiction. The former case gives us that $(i, j) = (k, l)$, so $e_{i, j}$ and $e_{j, k}$ denote the same edge. \qed
\end{proof}

We now prove that in an optimal translocation sequence $S$, the translocation can only denote edges between adjacent nodes.
\begin{lemma}
    \label{lem:translocation-is-edge}
    If $S = \{s_i\}$ is an optimal translocation sequence, then $s_i = (f_k, e_{k, j}) \vdash (f_j,r_k)$ for some $1 \leq k, j \leq N$ such that $(k, j) \in E$.
\end{lemma}
\begin{proof}
    Since $|B - A| = 2(N - 1)$, for $S$ to have size $N - 1$, every $s \in S$ must always introduce two new elements of $B$. If we write $s_i = (x_i, y_i) \vdash (u_i, v_i)$, we have $u_i, v_i \in \{f_{i'}, r_{i'} | 1 \leq {i'} \leq N \}$. We consider the following cases:

    \begin{enumerate}
        \item $u_i = f_j, v_i = f_k$: Since $f_j + f_k < 2(b_{2N} + 1)$, and $r_l, e_{l, m} \geq 3(b_{2N} + 1)$ 
        for all $1 \leq l, m \leq N$, then $x_i = f_l, y_i = f_m$ for some $1 \leq l,m \leq N$. Since $v$ is a $B_3$ sequence, the only possibility is $j = l$ and $k = m$. This is a contradiction, because translocations need to introduce new elements of $B$.
        \item $u_i = r_j, v_i = r_k$: If $S$ introduces two values $r_j$, since $B$ contains $N - 1$ values $f_j$ and $N - 1$ values $r_j$, $S$ introduces two values $f_j$ through some other translocation, which we proved at the previous point that it is not possible.
        \item $u_i = f_j, v_i = r_k$: We have $M \leq f_j + r_k < 2M$, so $x_i = f_l$ for some $1 \leq l \leq N$ and, either $y_i = r_m$ for some $1 \leq m \leq N$ or $y_i = e_{m,h}$ for some $1 \leq m,h \leq N$.
        Consider the first case, $y_i = r_m$: $f_j + r_k = f_l + r_m$ which is equivalent to $b_j + M + b_{N + k} = b_l + M + b_{N + m} $. After subtracting $M$ on both sides, since $b$ is a $B_3$ sequence, we have $j = l$ and $N + k = N + m$, which gives us $k = m$, so $s = (f_j, r_k) \vdash (f_j, r_k)$, which is not valid for an optimal $S$.
        
        Consider the other case $y_i = e_{m, h}$: $f_j + r_k = f_l + e_{m, h}$. Expanding $e_{m, h}$ gives us  $f_j + r_k = f_l + f_h + r_m - f_m$. By adding $f_m$ on both sides and replacing $f$ and $r$ with the values of $b$, we get: $b_m + b_j + M + b_{N + k} = b_l + b_h + M + b_{N + m}$. Since $b$ is a $B_3$ sequence, we have $k = m$ and, either $m = l$, $j = h$, or $m = h$, $j = l$:
        \begin{enumerate}
            \item Take $m = l$ and $j = h$, then $s_i = (f_k, e_{k, j}) \vdash (f_j, r_k)$. By Proposition~\ref{prop:all-distinct}, we have $e_{k, j}$ are different from $f_i,r_i$ for any $1 \leq i \leq N$, thus no translocation can introduce new values $e_{k, j}$. Therefore, $e_{k, j} \in A$, which implies $(k, j) \in E$.
            \item Take $m = h$, $j = l$, then $s_i = (f_j, e_{k,k}) \vdash (f_j, r_k)$. The graph does not contain edges from a node to itself, this we get a contradiction.
        \end{enumerate}
    \end{enumerate}
    \qed
\end{proof}

The following lemma states that an optimal translocation sequence describes a path in $G$.

\begin{lemma}
\label{lem:translocation-path}
 There is a sequence $(a_i)_{i = 1}^N$ such that $A_i(S) = E^* \cup \{f_{a_j} | 1 \leq j \le i + 1\} \cup \{r_{a_j} | 1 \leq j \leq i\}$ for $0 \leq i < N$, where $E^* = \{e_{i, j}| (i, j) \in E\}$. Moreover, if $s_i = (x_i, y_i) \vdash (u_i, v_i)$, then $x_i = f_{a_i}$, $u_i = f_{a_{i + 1}}$, and the sequence $(a_i)_{i = 1} ^ N$ is a path in $G$.
\end{lemma}

\begin{proof}
We show the lemma by induction. For $i = 0$ we have $A_0(S) = A = E \cup \{f_1\}$. Assuming the property holds for $i - 1$, we show it holds for $i$. Let $J_1 = \{j | f_j \in A_{i - 1}(S)\}$ and $J_2 = \{j | r_j \in A_{i - 1}(S)\}$. By the induction hypothesis, we have $J_1 = J_2 \cup \{a_{i}\}$. Due to the optimality of $S$, the $i$-th translocation must take as $v_i$ a value $r_k$ not present in $A_{i - 1}(S)$, but $x_i = f_k$ must be present in $A_{i - 1}(S)$, so the only solution is $k = a_{i}$.
     Since, by Lemma~\ref{lem:translocation-is-edge}, $s_i$ has form $(f_k, e_{k, j}) \vdash (f_j, r_k)$ where there is an edge between $k$ and $j$, then $j$ must be some node which has an edge from $k$, therefore, set $a_{i + 1} = j$ (which means $u_i = f_{a_{i + 1}}$). Thus, we have $A_{i}(S) = A_{i - 1}(S) \cup \{f_{a_{i + 1}}, r_{a_i}\}$, which completes the induction. \qed
\end{proof}

\begin{lemma}
    \label{lem:hamilton-path}
    Any path $(a_i)_{i = 1}^N$ in Lemma~\ref{lem:translocation-path} is a Hamiltonian path from $1$ to $N$.
\end{lemma}

\begin{proof}
  First, we know that $A = E^* \cup \{f_1\}$, so $a_1 = 1$. Due to the optimality of $S$, the  value $v_i = r_{a_i}$ cannot occur twice, so all $a_i$ for $1 \leq i < N$ are unique. Finally, $r_N \notin B$, so $a_i \neq N$ for $i < N$, and, since $f_N \in B$, the value $f_N$ must be added eventually, which can happen only on the last position, i.e. $u_N = f_N$ and $a_N = N$. \qed
\end{proof}
    
\begin{theorem}
    Deciding if $\mathit{TD}(A, B) = k$ for some $k$ is strongly NP-hard. 
\end{theorem}
\begin{proof}
    Lemma~\ref{lem:path-implies-td} together with Lemmas~\ref{lem:translocation-path} and~\ref{lem:hamilton-path} show there is a reduction from the Hamiltonian path problem to finding $\mathit{TD}(A, B)$ in a class of inputs where the values in $A$ and $B$ are polynomial in the size of the graph, $N$, which is polynomial in the size of the sets $A$ and $B$. Thus, the theorem follows. \qed
\end{proof}

\section{Exact Algorithm for the Case $|B| = 1$}

We next present several simple lemmas that will be used throughout the paper. As an illustrative application, we use these lemmas to derive a polynomial-time algorithm for the case $|B| = 1$.

Let $B = \{z\}$ and $x=\max A$. The algorithm contains 2 cases: if $x > z$ then the answer is 1, otherwise if $z > x$ the answer is $\lceil \log_2 \frac{z}{x} \rceil$.  

\begin{lemma}\label{lemma:interval}
Let $p,q\in\mathbb{N}$. For every integer $t$ with $0\le t\le p+q$, there exist integers $i\in[0,p]$ and $j\in[0,q]$ such that $(p,q)\ \vdash_{i,j}\ (t,\ p+q-t)$.
\end{lemma}

\begin{proof}
Recall that applying a unary translocation with parameters $i$ and $j$ to $(p,q)$ produces the pair $(p-i+j,\ q-j+i)$. We want the first produced value to be equal to $t$, that is, $p-i+j=t$. Equivalently, we need $j-i=t-p$.
Let $d=t-p$. Since $0\le t\le p+q$, we have $-p\le d\le q$. If $d\geq 0$, choose $i=0$ and $j=d$. Then $0\le j\le q$, and $p-i+j=p+d=t$. If $d<0$, choose $i=-d$ and $j=0$. Then $0\le i\le p$, and $p-i+j=p+d=t$.

In both cases, the first produced value is $t$. Since unary translocations preserve the total sum, the second produced value is $p+q-t$. Therefore, $(p,q)\ \vdash_{i,j}\ (t,\ p+q-t)$. \qed
\end{proof}

\begin{corollary}\label{cor:aa-interval}
For any $a\in\mathbb{N}$ and any $t\in[0, 2 \cdot a]$, a single translocation from $(a, a)$ produces $(t, 2 \cdot a-t)$.
\end{corollary}

\begin{proof}
Apply Lemma~\ref{lemma:interval} with $p=q=a$. \qed
\end{proof}

\begin{lemma}\label{lemma:doubling}
    Let $M_k$ be the maximum available length (the maximum value that can be obtained through translocations) after $k$ translocations. Then, $M_k \leq x \cdot 2^k$, where $x = \max A$.
\end{lemma}

\begin{proof}
    We prove the lemma by induction on the number of translocations. \\
    \noindent $P(k): M_k \leq x \cdot 2^k$, $\forall k \in \mathbb{N}$. \\
    \noindent \textbf{Base case}. $P(0)$ is true since $x$ is the largest length in $A$. \\
    \textbf{Induction step}. $P(k) \implies P(k + 1)$. Let $a$ be the largest value obtained after $k$ translocations. The largest length available after $k + 1$ steps is $2\cdot a$ obtained by translocation $(a, a) \vdash_{0, a} (2\cdot a, 0)$. Using the induction hypothesis we have $2 \cdot a \leq 2 \cdot x \cdot 2^k \leq x \cdot 2 ^ {k + 1}$. \qed
\end{proof}

\begin{corollary}\label{cor:lower}
     If $z > x$, any sequence producing $z$ requires at least $\lceil \log_2 \frac{z}{x} \rceil$.
\end{corollary}

\begin{proof}
    If $t < \bigl \lceil \log_2 \frac{z}{x} \bigr\rceil $ then $x \cdot 2^t < z$. By Lemma~\ref{lemma:doubling}, $M_t \leq x \cdot 2^t < z$, so $z$ is unattainable in $t$ steps. \qed
\end{proof}

\begin{lemma}\label{lemma:upper}
Suppose $z>x$ and let $t = \bigl\lceil \log_2 \tfrac{z}{x} \bigr\rceil$. Then $z$ can be produced in at most $t$ translocations.
\end{lemma}

\begin{proof}
Perform $t-1$ doublings starting from $a_0=x$ using $(a_i,a_i)\vdash_{0,a_i}(a_{i+1} = 2 \cdot a_i,0)$. After $t-1$ steps we have $a_{t-1}=x\cdot 2^{t-1}$. By the choice of $t$,
$a_{t-1}\le z\le 2 \cdot a_{t-1}$. By Corollary~\ref{cor:aa-interval}, one translocation from $(a_{t-1},a_{t-1})$ yields $(z,\,2 \cdot a_{t-1}-z)$, so $z$ appears within $t$ steps. \qed
\end{proof}

\begin{corollary}\label{cor:upper}
If $z>x$, then $\mathit{TD}(A,\{z\}) \le \bigl\lceil \log_2 \tfrac{z}{x} \bigr\rceil$. If $ z < x$, then $\mathit{TD}(A,\{z\}) \le 1$.
\end{corollary}

\begin{proof}
The first inequality follows from Lemma~\ref{lemma:upper}. If $z < x \implies z < x + x$, apply Lemma~\ref{lemma:interval} with $(p,q)=(x,x)$ and obtain $(z,2 \cdot x -z)$ in one step. \qed
\end{proof}

From Lemma~\ref{lemma:interval} and  Corollary~\ref{cor:lower},~\ref{cor:upper} it follows the next theorem. 
\begin{theorem}\label{thm:b1}
Let $x=\max A$. Then, $\mathit{TD}(A,\{z\}) = 1$ if $z < x$ and $\bigl\lceil \log_{2}\frac{z}{x}\bigr\rceil$ if $z > x$. Moreover, $\mathit{TD}(A,\{z\})$ can be computed in $O(|A|)$ time.
\end{theorem}

\section{Approximation Algorithms}\label{sec:approx}

\subsection{2-approximation}

The 2-approximation algorithm has two steps. First, we design a greedy algorithm that is optimal in a simplified model where each
translocation has exactly one designated output (see Definition~\ref{def:1-translocation} of the \emph{1-translocation model}).  
In this algorithm, we sort the targets $B = \{b_1, b_2, \dots, b_r\}$ increasingly and process them one by one.
We maintain a current maximum available length $x$, initially $x = \max A$.
For each target $b_i$:

\begin{enumerate}
    \item if $x > b_i$, we perform one step to produce $b_i$ from $(x,x)$;
      in this case the maximum remains $x$.
    \item if $x \le b_i$, we start with $a = x$ and repeatedly double it ($(a, a)\vdash(2a, 0)$) until we are close to $b_i (b_i \leq 2 \cdot a)$. Then, we produce $b_i$ from $(a, a)$ by adjusting the output values (instead of another doubling) and update maximum ($x = b_i$).
\end{enumerate}

We show that this greedy strategy is optimal in the 1-translocation model.
Then, using Theorem~\ref{thm:aprox}, we obtain a 2-approximation for the original model, where both outputs of each translocation are counted. The algorithm is formally presented in Algorithm~\ref{alg:two-approx}.

\begin{definition}[1-translocation]
\label{def:1-translocation}
In the unary setting, a \emph{1-translocation} is a translocation in which we only
designate a single output length. Formally, given a translocation
$(x,y)\vdash_{i,j}(u,v),$
a 1-translocation chooses one of $u$ or $v$ as its output and is written $(x,y)\vdash^{1} w$ with $w\in\{u,v\}$.
\end{definition}

We first relate the cost in the original model to the cost in
the 1-translocation model.

\begin{theorem}\label{thm:aprox}
    Let $OPT$ be the minimum number of translocations in any sequence that
    produces the desired targets, and let $OPT_1$ be the minimum number of
    1-translocations in any 1-translocation sequence that produces the same targets.
    Then $ OPT \;\ge\; \frac{OPT_1}{2}.$
\end{theorem}

\begin{proof}
Let $S = \{s_1, s_2, \dots, s_{OPT}\}$ be an optimal sequence of  translocations,
where each $s_i = (x_i, y_i) \vdash (u_i, v_i), \quad 1 \leq i \leq OPT$,
is a translocation, and let
$S_1 = \{s^1_1, s^1_2, \dots ,s^1_{OPT_1}\}$
be an optimal sequence of 1-translocations, so $|S| = OPT$ and $|S_1| = OPT_1$. Suppose, for the sake of contradiction, that $OPT < \frac{OPT_1}{2}$.
From the translocation solution $S$ we construct a 1-translocation
solution $S_2$ as follows. For each $i=1,\dots,OPT$, replace $s_i = (x_i, y_i) \vdash (u_i, v_i)$
by two 1-translocations
$
    s^2_{2i} = (x_i, y_i) \vdash^{1} u_i
    \quad\text{and}\quad
    s^2_{2i+1} = (x_i, y_i) \vdash^{1} v_i.
$
This yields a valid 1-translocation sequence $S_2 = \{s^2_1, s^2_2, \dots, s^2_{2\cdot OPT}\},$
because each translocation in $S$ produces both $u_i$ and $v_i$, and in the 1-translocation model we are simply designating these two outputs separately in two steps.
By construction, $|S_2| = 2\cdot OPT < OPT_1 = |S_1|$, which contradicts the optimality of $S_1$ in the 1-translocation model. Therefore our assumption
was false and we must have $OPT \ge \frac{OPT_1}{2}$. \qed
\end{proof}

\begin{algorithm}[thb]
\caption{Compute a 2-approximation to $\mathit{TD}(A,B)$}
\label{alg:two-approx}
\begin{algorithmic}[1]
\State $x \gets \max A$ 
\State sort $B$ increasingly as $b_1\le\cdots\le b_r$
\State $s \gets 0$ \Comment{number of 1-translocations}
    \For{$i=1$ \textbf{to} $r$}
        \If{$x > b_i$}
            \State use one 1-translocation from $(x,x)$ to output $b_i$ \Comment{use Lemma~\ref{lemma:interval} with $(p,q)=(x,x)$ ($b_i < x$ and $x \in \mathbb{N} \implies b_i < x + x $)}
            \State $s \gets s + 1$
            \Comment{$x$ remains the maximum}
        \Else
            \State $d \gets \left\lceil \log_{2} \frac{b_i}{x} \right\rceil$
            \State perform $d-1$ doublings $(a,a)\vdash^{1} 2a$ starting from $a=x$ to reach $a=x\cdot 2^{d-1}$
            \State from $(a,a)$ use Lemma~\ref{lemma:interval} to output $b_i$ in one 1-translocation
            \State $s \gets s + d$
            \State $x \gets b_i$ \Comment{new maximum is at least $b_i$}
        \EndIf
    \EndFor
\State \Return{$s$}
\end{algorithmic}
\end{algorithm}

We now analyze optimality in the 1-translocation model.

\begin{lemma}
There exists an optimal sequence of 1-translocations that produces the targets
in increasing order, i.e., $b_1 \leq b_2 \leq \dots \leq b_r$.
\end{lemma}

\begin{proof}
A single $1$-translocation applied to two numbers $p$ and $q$ can produce any
$k \in \{0,1,\dots,p+q\}$. Consider an optimal sequence of $1$-translocations
that produces exactly the set $\{b_1,\dots,b_r\}$. Among all such optimal
sequences, choose one for which the list of produced targets has the minimum
number of inversions (an inversion is a pair of indices $i<j$ with $b_i>b_j$).
If there is an inversion, then there exists an adjacent one, i.e., an index
$i$ with $b_i>b_{i+1}$. Let $M$ be the largest number present immediately
before the two consecutive translocations that output $b_i$ and then $b_{i+1}$.
Since $b_i$ is produced at that time using numbers at most $M$, we have
$b_i\le 2M$, hence also $b_{i+1}\le 2M$, so $b_{i+1}$ could be produced instead
in the first of these two steps. Swapping the order of these two target outputs
yields another optimal sequence producing the same target set but with fewer
inversions, contradicting minimality. Therefore, there exists an optimal
sequence with no inversions, i.e., with $b_1 \le b_2 \le \dots \le b_r$.
\qed
\end{proof}

By the lemma above, fix an optimal $1$-translocation sequence that produces the
targets in nondecreasing order. Let $t_i$ be the index of the first
$1$-translocation whose output equals $b_i$ ($t_0=0$). Define the block
$(t_{i-1},t_i] = \{t_{i-1}+1, t_{i-1}+2, \dots, t_i\}$, and denote by $M_j$ the
maximum available length after the $j$-th $1$-translocation in this sequence.

\begin{lemma}\label{lemma:block}
There exists an optimal $1$-translocation sequence such that for every
$i=1,\dots,r$:
\begin{enumerate}
\item Every $1$-translocation in the block $(t_{i-1},t_i)$ is \emph{necessary}
to make $b_i$ producible at step $t_i$ (in the sense that removing it, while
keeping the relative order of the remaining steps in $(t_{i-1},t_i]$, makes
$b_i$ impossible at step $t_i$).
\item $M_{t_{i-1}}=\max\{x,b_{i-1}\}$ (with the convention $b_0 =0$, so
$M_{t_0}=x$).
\end{enumerate}
\end{lemma}

(1) Fix $i$. Call a step $p\in(t_{i-1},t_i)$ \emph{unnecessary (for $b_i$)} if,
after deleting $p$ and keeping the other steps in $(t_{i-1},t_i]$ in the same
relative order, $b_i$ is still producible at step $t_i$. If such a step exists,
move it immediately after $t_i$. This does not affect feasibility up to
$t_i$ (by definition of unnecessary), and it cannot harm later feasibility
because $1$-translocations do not remove inputs. Repeating this inside each
block yields an optimal sequence in which every step in $(t_{i-1},t_i)$ is
necessary.

(2) We already have $M_{t_{i-1}}\ge \max\{x,b_{i-1}\}$ because $x\in A$ and
$b_{i-1}$ is produced by time $t_{i-1}$. It remains to show
$M_{t_{i-1}}\le \max\{x,b_{i-1}\}$.

Assume for contradiction that $M_{t_{i-1}}>\max\{x,b_{i-1}\}$, and let $p\le t_{i-1}$
be the \emph{first} step at which the maximum becomes larger than
$\max\{x,b_{i-1}\}$. Let $M$ be the maximum just before step $p$, and let the
output of step $p$ be some value $w$ with $w>\max\{x,b_{i-1}\}$. Since a
$1$-translocation from two inputs at most $M$ can output only values $\le 2M$,
we have $w\le 2M$, hence $M \ge \frac{w}{2} > \frac{\max\{x,b_{i-1}\}}{2} \ge \frac{b_{i-1}}{2}$.
Therefore $b_{i-1} < 2M$, meaning that $b_{i-1}$ was already producible at step
$p$ from the pair $(M,M)$. Consequently, step $p$ is not needed in order to
eventually produce $b_{i-1}$, and in particular it is unnecessary for the block
that contains $t_{i-1}$ contradicting (1), which eliminates all such unnecessary steps in each block. Hence the assumption was false and $M_{t_{i-1}}\le \max\{x,b_{i-1}\}$. Thus, $M_{t_{i-1}} = \max\{x,b_{i-1}\}$.

We now bound the number of 1-translocations needed in each block.

\begin{lemma}
Let $\ell_i = t_i - t_{i-1}$ be the length of block $(t_{i-1},t_i]$ in an optimal
1-translocation sequence as in Lemma~\ref{lemma:block}. Then:
\begin{enumerate}
    \item For $i=1$, we have $\ell_1 \ge 1$.
    \item For every $i\ge 2$,
    $
    \ell_i \geq \Bigl\lceil \log_{2} \frac{b_i}{b_{i-1}} \Bigr\rceil.
    $
\end{enumerate}
\end{lemma}

\begin{proof}
For $i=1$, at least one $1$-translocation is needed to produce $b_1\notin A$, so $\ell_1\ge 1$.
For $i\ge 2$, Lemma~\ref{lemma:block} gives $M_{t_{i-1}}=b_{i-1}$.
In $\ell_i$ steps after $t_{i-1}$, the maximum can grow by at most a factor $2^{\ell_i}$, because each step can increase the maximum by at most a factor $2$ (Lemma~\ref{lemma:doubling}).
Moreover, the final step of the block (which outputs $b_i$) can output at most twice the maximum available length just before that step, i.e., at time $t_i-1$.
That is,
$
b_i \le 2\cdot M_{t_i-1} \le 2^{\ell_i}\cdot M_{t_{i-1}} = 2^{\ell_i}\cdot b_{i-1}.
$
Therefore, $\ell_i \ge \log_2\frac{b_i}{b_{i-1}}$, and since $\ell_i$ is integral,
$
\ell_i \ge \Bigl\lceil \log_2\frac{b_i}{b_{i-1}} \Bigr\rceil.
$

\qed
\end{proof}

\begin{lemma}
Algorithm~\ref{alg:two-approx} produces $b_1,\dots,b_r$
using an optimal number of $1$-translocations. 
\end{lemma}

\begin{proof}
Consider the optimal $1$-translocation sequence from Lemma~\ref{lemma:block}, and
think of it as divided into blocks that first produce $b_1$, then $b_2$, and so on.
We compare Algorithm~\ref{alg:two-approx} to this sequence. For $b_1$:
\begin{enumerate}
    \item If $b_1 < x$, the optimal sequence must use at least one $1$-translocation to introduce $b_1$ (since $b_1\notin A$). Algorithm~\ref{alg:two-approx} also uses exactly one step to output $b_1$ from $(x,x)$, so it is optimal in this case.

    \item If $b_1 \ge x$, the optimal sequence starts from maximum $x$ (Lemma~\ref{lemma:block})
     and must reach $b_1$ by a sequence of $1$-translocations. By the same doubling argument as in the case $|B| = 1$, at least
    $d_1 = \bigl\lceil \log_2 \tfrac{b_1}{x} \bigr\rceil$ steps are needed, and
  Algorithm~\ref{alg:two-approx} uses exactly $d_1$ steps for $b_1$.
\end{enumerate}

For $i\ge 2$, the optimal sequence of Lemma~\ref{lemma:block} starts block $i$ with
maximum $b_{i-1}$ and needs at least
$
\Bigl\lceil \log_2 \frac{b_i}{b_{i-1}} \Bigr\rceil
$
steps to produce $b_i$.

Algorithm~\ref{alg:two-approx} maintains the invariant that $x$ is the current
maximum available length and $x \ge b_{i-1}$ at the start of the iteration for $b_i$.
Indeed, after processing $b_{i-1}$ we either had $x>b_{i-1}$ and produced $b_{i-1}$
in one step without changing $x$, or we had $x\le b_{i-1}$ and repeatedly doubled
$x$ until reaching $b_{i-1}$ (ending with $x=b_{i-1}$). Since $x$ is never decreased,
the invariant follows. Thus:
\begin{enumerate}
    \item If $x > b_i$, then any sequence must perform at least one step to output $b_i$, and the algorithm uses exactly one step.
    \item If $x \le b_i$, the algorithm uses
    $
    d_i = \Bigl\lceil \log_2 \frac{b_i}{x} \Bigr\rceil
    \le
    \Bigl\lceil \log_2 \frac{b_i}{b_{i-1}} \Bigr\rceil
    $
    steps to reach $b_i$, which is no more than the blockwise lower bound for any
    optimal sequence starting from $b_{i-1}$.
\end{enumerate}

Summing over all $i$, Algorithm~\ref{alg:two-approx} uses no more $1$-translocations
than this optimal sequence. Since that sequence is optimal by definition, the
algorithm cannot use strictly fewer steps. Hence it uses exactly the optimal number
$OPT_1$ of $1$-translocations.

\qed
\end{proof}

\begin{theorem}\label{thm:unary-2approx}
Algorithm~\ref{alg:two-approx} is a $2$-approximation for $\mathit{TD}(A, B)$. It runs in $O(|A|+|B|\log|B|)$ time to compute the distance (and in $O(|A|+|B|\log|B|+s)$ time if it outputs the explicit translocation sequence, where $s$ is the number of 1-translocations).
\end{theorem}

\begin{proof}
Let $OPT$ be the optimal number of unary translocations needed to produce $B$ from $A$ in the non-contiguous model, i.e., $OPT = \mathit{TD}(A,B)$, and let $OPT_1$
be the optimal number of 1-translocations for the corresponding restricted model. By Theorem~\ref{thm:aprox}, for any instance we have $OPT \ge \frac{OPT_1}{2}$. The
previous lemma shows that Algorithm~\ref{alg:two-approx} computes an optimal
1-translocation sequence, so it uses exactly $OPT_1$ steps. Therefore,
$
\text{cost}(\text{Algorithm~\ref{alg:two-approx}}) = OPT_1 \leq 2 \cdot OPT
= 2\cdot \mathit{TD}(A,B),
$
which proves the 2-approximation guarantee. For the running time, computing $x=\max A$ takes $O(|A|)$, sorting $B$ takes
$O(|B|\log|B|)$, and the for-loop runs in $O(|B|)$ iterations, each performing
only $O(1)$ operations to accumulate the value $s$. If we additionally output the
explicit sequence of 1-translocations,
we need $O(s)$ extra work to emit them. Hence the total time is
$O(|A|+|B|\log|B|)$ for the distance, and $O(|A|+|B|\log|B|+s)$ with an
explicit sequence. \qed
\end{proof}

\subsection{Additive $|B| - 1$ Approximation}

In this subsection, we present an additive $(|B|-1)$-approximation algorithm for the unary translocation distance. We then show that the same algorithm also yields a multiplicative $3$-approximation.

Let $x = \max A$ and $z_{\max} = \max B$. The algorithm has two phases:
\begin{enumerate}
    \item First, produce $z_{\max}$ from $A$ using exactly the same logic as in the case $|B| = 1$ (Theorem~\ref{thm:b1}).
    \item Once $z_{\max}$ is available, obtain each remaining $z \in B \setminus \{z_{\max}\}$ in one additional translocation from $(z_{\max},z_{\max})$.
\end{enumerate}

\begin{lemma}\label{lemma:lb-sets}
Let $x=\max A$ and $z_{\max} = \max B$. Then, $\mathit{TD}(A,B) \geq \mathit{TD}(A,\{z_{\max}\})$.
\end{lemma}

\begin{proof}
Any translocation sequence that obtains the target set $B$ must in particular produce $z_{\max}$, so its length is at least
$\mathit{TD}(A,\{z_{\max}\})$. The explicit expression for this value follows from
Theorem~\ref{thm:b1}. \qed
\end{proof}

\begin{corollary}\label{cor:onestep-from-zmax}
If $z_{\max}$ is available (i.e., appears in some translocation sequence), then for any $z \le z_{\max}$ one translocation yields a length $z$.
\end{corollary}

\begin{proof}
Once $z_{\max}$ has appeared at least once, in the non-contiguous model we may use the pair
$(z_{\max}, z_{\max})$ in a translocation. By Corollary~\ref{cor:aa-interval} (with $a=z_{\max}$),
for any $t \in [0,2z_{\max}]$ a single translocation from $(z_{\max},z_{\max})$ produces
$(t, 2z_{\max}-t)$. In particular, for any $z \le z_{\max}$ we can choose $t=z$ and obtain $z$
in one step. \qed
\end{proof}

\begin{lemma}\label{lemma:ub-sets}
Let $z_{\max} = \max B$ and $ t = \mathit{TD}(A,\{z_{\max}\})$.
There exists a sequence of at most $ t + (|B|-1)$
translocations that generates the target set $B$.
\end{lemma}

\begin{proof}
By definition of $t$, there exists a translocation sequence $S_z$ of length $t$ that produces $z_{max}$.
Let $S_z$ end with the translocation $s_t$. After executing $S_z$, $z_{\max}$ is available. For each $z \in B \setminus \{z_{\max}\}$, apply Corollary~\ref{cor:onestep-from-zmax} once,
starting from the current set of available lengths. Each such translocation
adds $z$ to the set of produced lengths. Because $B$ is a set, there are exactly $|B|-1$ such distinct $z$, so we need at most $|B|-1$ additional steps. Thus, we obtain a $B$-producing translocation sequence of total length at most $t + (|B|-1)$. \qed
\end{proof}

\begin{theorem}\label{thm:additive-Bminus1}
Let $x=\max A$, $z_{\max}=\max B$, and $t = \mathit{TD}(A,\{z_{\max}\})$.
Then, $ t \leq \mathit{TD}(A,B) \leq t + |B| - 1$.
Consequently, $\mathit{TD}(A,B)$ can be approximated with $(|B|-1)$-additive factor in $O(|A|+|B|)$ time. \qed
\end{theorem}

\begin{proof}
The lower bound $t \le \mathit{TD}(A,B)$ follows from Lemma~\ref{lemma:lb-sets}. The upper bound
$\mathit{TD}(A,B) \le t + |B|-1$ follows from Lemma~\ref{lemma:ub-sets}. Computing
$x=\max A$ and $z_{\max}=\max B$ takes $O(|A|+|B|)$ time, and evaluating $t$ and
$t+(|B|-1)$ is constant-time, so the overall running time is $O(|A|+|B|)$. \qed
\end{proof}

\begin{corollary}\label{cor:additive-implies-3approx}
The algorithm from Theorem~\ref{thm:additive-Bminus1} is a $3$-approximation for the unary translocation distance problem.
\end{corollary}

\begin{proof}
Let $\mathrm{OPT}=\mathit{TD}(A,B)$, and let $U$ be the number of translocations used by the algorithm from Theorem~\ref{thm:additive-Bminus1}. By Theorem~\ref{thm:additive-Bminus1}, we have $U\leq t+|B|-1$, where $t=\mathit{TD}(A,\{z_{\max}\})$. By Lemma~\ref{lemma:lb-sets}, $t\leq \mathrm{OPT}$. Moreover, since $A\cap B=\emptyset$, every element of $B$ has to be produced by some translocation. Each translocation can introduce at most two new target lengths, and therefore $\mathrm{OPT}\geq \lceil \frac{|B|}{2} \rceil$. Hence $|B|-1\leq 2\mathrm{OPT}-1\leq 2\mathrm{OPT}$. Combining these inequalities, we obtain $U\leq t+|B|-1\leq \mathrm{OPT}+2\mathrm{OPT}=3\mathrm{OPT}$. Therefore, the algorithm is a $3$-approximation.

\qed
\end{proof}

\section{Exact Algorithm for the Case $|B|$ = 2}\label{sec:exact}

Next we consider the case $B=\{z_1,z_2\}$ with $z_1 < z_2$ and $x=\max A$. Informally, the algorithm proceeds as follows:

\begin{enumerate}
    \item First, check whether we can obtain both $z_1$ and $z_2$ in a single translocation
    from the initial set $A$. This happens exactly when there exist $a,b\in A$ such that
    $a+b=z_1+z_2$. By Lemma~\ref{lemma:interval}, one step then transforms $(a,b)$ into
    $(z_1,z_2)$, so the answer is $1$.

    \item Otherwise, we know that at least $2$ translocations are needed. We then distinguish
    cases according to the relative position of $z_2$ and $x=\max A$:
    \begin{enumerate}
        \item If $z_2 < x$, then by Theorem~\ref{thm:b1} we have
        $\mathit{TD}(A,\{z_2\}) = 1$. Since we have already ruled out the possibility of
        producing $\{z_1,z_2\}$ in a single step, it follows that
        $\mathit{TD}(A,\{z_1,z_2\}) = 2$.

        \item If $z_2 > x$, then
        $\mathit{TD}(A,\{z_2\}) =
        \left\lceil \log_{2}\frac{z_2}{x} \right\rceil$. Let $t^* = \mathit{TD}(A,\{z_2\}) = \left\lceil \log_{2}\frac{z_2}{x} \right\rceil$,
        by Theorem~\ref{thm:additive-Bminus1} with $|B|=2$, we know that
        $\mathit{TD}(A,\{z_1,z_2\}) \in \{t^*, t^*+1\}$. We test whether we can
        obtain $z_1$ during the doubling construction for $z_2$ or obtain both $z_1$ and $z_2$ in the last step, i.e.,
        without increasing the total number of steps beyond $t^*$. If this is feasible,
        the answer is $t^*$, otherwise it is $t^*+1$.
    \end{enumerate}
\end{enumerate}

We now formalize the bounds and the feasibility condition used above.

\begin{lemma}\label{lemma:two-bounds}
Let $B=\{z_1,z_2\}$ with $z_1 < z_2$, $x=\max A$ and $t^* = \mathit{TD}(A,\{z_2\}) = 1 \text{ if } z_2 < x  \text{ else } \bigl\lceil \log_{2}\tfrac{z_2}{x}\bigr\rceil$. Then, $ t^* \leq \mathit{TD}(A,\{z_1,z_2\}) \leq t^* + 1.$
\end{lemma}

\begin{proof}
Any sequence that produces $\{z_1,z_2\}$ produces $z_2$, hence
$
\mathit{TD}(A,\{z_1,z_2\}) \;\ge\; \mathit{TD}(A,\{z_2\}) = t^*.
$\
For the upper bound, apply Theorem~\ref{thm:additive-Bminus1} with $|B|=2$ and
$z_{\max}=z_2$, which yields $\mathit{TD}(A,\{z_1,z_2\}) \le t^* + (2-1) = t^*+1$. \qed
\end{proof}

The next lemma characterizes when the lower bound $t^*$ from Lemma~\ref{lemma:two-bounds}
is actually attained in the regime $z_2 > x$ and $t^*\ge 2$.

\begin{lemma}\label{lemma:feasibility}
Assume $z_2 > x$ and
$ t^* = \bigl\lceil \log_{2}\tfrac{z_2}{x}\bigr\rceil \ge 2$. Let $
a_k = x\cdot 2^k, \forall k \in \{ 0,1,\dots,t^*-1\}$.
Then
$
\mathit{TD}(A,\{z_1,z_2\}) = t^* \iff$ at least one of the following statements is true:
\begin{enumerate}
    \item $\exists\,k\in\{0,\dots,t^*-2\}:
\ z_1 \le 2a_k\ \text{ and }\
\max\{z_1,\,2a_k-z_1\}\cdot 2^{\,t^*-1-k} \ \ge\ z_2$
    \item $z_1 + z_2$ is an even number and $z_1 + z_2 \leq x\cdot 2^{ t^*}$\
    \item $z_1 + z_2$ is an odd number and $z_1 + z_2 \leq x\cdot 2^{ t^*-1} + x\cdot 2^{ t^*-2} $. 
\end{enumerate}
 
\end{lemma}

\begin{proof}
($\Rightarrow$) Suppose $\mathit{TD}(A,\{z_1,z_2\}) = t^*$. Then $\exists S = \{s_1 = (x_1, y_1) \vdash (u_1, v_1),  s_2 = (x_2, y_2) \vdash (u_2, v_2), \dots, s_{t^*} =  (x_{t^*}, y_{t^*}) \vdash (u_{t^*}, v_{t^*})\}$ such that $|S| = t^*$ and $S$ is a translocation sequence that produces $\{z_1,z_2\} \implies \exists i, j, 1 \leq i, j \leq t^*$ such that $z_1 \in \{u_i, v_i\}$ and $z_2 \in \{u_j, v_j\}$. Since $\mathit{TD}(A, \{z_2\}) = t^*$, $z_2$ cannot be obtained before step $t^* \implies j = t^* (z_2 \in \{u_{t^*}, v_{t^*}\})$.
For the first step $k+1, 0 \leq k < t^*$ that obtains $z_1$ ($k+1 = min\{i| 1 \leq i \leq t^* \text{ and } z_1 \in \{u_i, v_i\}\}$), we distinguish two cases:
\begin{itemize}
    \item $k+1 < t^*$, then $k\in\{0,\dots,t^*-2\}$ and $z_1 \in \{u_{k+1}, v_{k+1}\}$. WLOG, assume $z_1 = u_{k+1}$. By Lemma~\ref{lemma:doubling}, the largest obtainable value at step $k$ is $x\cdot 2^k = a_k$, then $\max A_k \leq a_k$, $z_1 \leq 2 \cdot \max A_k \leq 2 \cdot a_k$ and $v_{k+1} \leq 2 \cdot \max A_k - z_1\leq 2 \cdot a_k - z_1$. Additionally,  $\max A_{k+1} = \max\{ \max A_k, z_1, v_{k+1}\}$ and we have ${t^*}-k-1$ remaining steps to reach $z_2$. The largest obtainable value after these ${t^*}-k-1$ steps is at most $\max A_{k+1}\cdot 2^{{t^*}-k-1}\implies \max A_{k+1}\cdot 2^{{t^*}-k-1}\geq z_2 > x \cdot 2^{t^*-1} \implies \max A_{k+1} > x \cdot 2^k = a_k$. Since $\max A_k \leq a_k$, $\max A_{k+1} = \max \{z_1, v_{k+1}\} \leq \max \{z_1, 2a_k - z_1\} \implies \max \{z_1, 2a_k - z_1\}\cdot 2^{{t^*}-k-1} \geq \max \{z_1, v_{k+1}\}\cdot 2^{{t^*}-k-1}\geq z_2$. Thus, statement $1$ holds.
    \item $k+1 = t^* \implies \{u_{t^*}, v_{t^*}\} = \{z_1, z_2\}$. By the sum of $z_1$ and $z_2$, we distinguish 2 sub-cases:
    \begin{itemize}
        \item $z_1 + z_2$ is even, then $z_1, z_2$ can be produced using doubling $\implies z_1 + z_2 \leq 2 \cdot \max A_k = 2 \cdot \max A_{t^* - 1}$. Since the largest value obtainable at step $t^* - 1$ is $ x \cdot 2^ { t^*-1}$, $\max A_{t^* - 1} \leq x \cdot 2^ { t^*-1}$ and $z_1 + z_2 \leq 2 \cdot x \cdot 2^ { t^*-1} = x \cdot 2^ { t^*}$. Thus, statement $2$ holds.
        \item $z_1 + z_2$ is odd, then $z_1, z_2$ can be produced only using two different lengths from $A_k = A_{t^* - 1}$. The largest value obtainable in $k = t^*-1$ steps is $x \cdot 2^ { t^*-1}$ ($\max A_{t^* - 1} \leq x \cdot 2^ { t^*-1}$). But only one of the values obtained at step $k$ can be greater than $\max A_{k-1}$ ($u_k + v_k \leq 2 \cdot \max A_{k-1}$. WLOG, assume $u_k > \max A_{k-1} \implies v_k < \max A_{k-1}$). Therefore, the sum of the two largest distinct values in $A_k$ is at most $\max A_k + \max A_{k-1} = x \cdot 2^ {k} + x \cdot 2^ { k-1} = x \cdot 2^ { t^*-1} + x \cdot 2^ { t^*-2} \implies z_1 + z_2 \leq x \cdot 2^ { t^*-1} + x \cdot 2^ { t^*-2}$ and the statement $3$ holds.
    \end{itemize}
\end{itemize}

($\Leftarrow$) Conversely, suppose that at least one of the following statements is true. We construct a sequence as follows:
\begin{enumerate}
    \item $\exists k\in\{0,\dots,t^*-2\}$ such that
    $z_1\le 2a_k$ and $\max\{z_1,2a_k-z_1\}\cdot 2^{t^*-1-k} \ge z_2$.
    
    For the first $k$ steps, follow the standard doubling
    sequence from $x$:
    $
    a_0=x,\ a_1=2x,\ \dots,\ a_k=x\cdot 2^k.
    $
    At step $k{+}1$, use Corollary~\ref{cor:aa-interval} on $(a_k,a_k)$ to produce
    $(z_1,2a_k-z_1)$. After this, let $b_0 = \max\{z_1,2a_k-z_1\}$. In the next
    $t^*-2-k$ steps, repeatedly apply doublings of the form
    $(b_i,b_i)\vdash_{b_i, 0}(0, 2b_i)$ to obtain $b_{t^*-2-k} = b_0\cdot 2^{t^*-2-k}$.
    Then, by the inequality $b_0\cdot 2^{t^*-1-k} = 2b_{t^*-2-k}\ge z_2$, we can use $(b_{t^*-2-k},b_{t^*-2-k})$ to produce $(z_2,2b_{t^*-2-k}-z_2)$. 

    \item $z_1 + z_2$ is an even number and $z_1 + z_2 \leq x\cdot 2^{ t^*}$.\\
    For the first $t^* - 2$ steps, follow the standard doubling
    sequence from $x$:
    $a_0=x,\ a_1=2x,\ \dots,\ a_{t^*-2}=x\cdot 2^{t^*-2}$.
    As $ z_1 + z_2 \leq x \cdot 2^{t^*} \implies (z_1+z_2)/2 \leq x \cdot 2^{t^*-1} \leq 2 \cdot x \cdot 2^{t^*-2}$, at step $t^*-1$ we use Corollary~\ref{cor:aa-interval} on $(a_{t^*-2},a_{t^*-2})$ to produce
    $(\frac{z_1+z_2}{2},2 a_{t^*-2}-\frac{z_1+z_2}{2})$. Then, at step $t^*$ we use $(\frac{z_1+z_2}{2},\frac{z_1+z_2}{2})$ and produce $(z_1, z_2)$.

    \item $z_1 + z_2$ is an odd number and $z_1 + z_2 \leq x\cdot 2^{ t^*-1} + x\cdot 2^{ t^*-2}$.\\
    For the first $t^* - 2$ steps, follow the standard doubling
    sequence from $x$:
    $a_0=x,\ a_1=2x,\ \dots,\ a_{t^*-2}=x\cdot 2^{t^*-2}$.
    Since $ \mathit{TD}(A, \{z_2\}) = t^*$, $z_2 > x \cdot 2^{t^*-1} \implies x \cdot 2^{t^*-1} < z_1+z_2 \leq x\cdot 2^{ t^*-1} + x\cdot 2^{ t^*-2} \implies x \cdot 2^{t^*-2} < (z_1+z_2) - x\cdot 2^{ t^*-2} \leq x\cdot 2^{ t^*-1} = 2a_{t^*-2}$. Therefore, at step $t^*-1$, we can use Corollary~\ref{cor:aa-interval} on $(a_{t^*-2},a_{t^*-2})$ to produce $((z_1+z_2) - x\cdot 2^{ t^*-2}, 2a_{ t^*-2} - ((z_1+z_2) - x\cdot 2^{ t^*-2}))$. Finally, at step $t^*$, we use $((z_1+z_2) - x\cdot 2^{ t^*-2}, x\cdot 2^{ t^*-2})$ and obtain $(z_1, z_2)$.
    
    \end{enumerate}
    For each of the cases considered above, we have obtained both $z_1$ and $z_2$ using $t^*$ translocations, so $\mathit{TD}(A,\{z_1,z_2\})\le t^*$. Combined with
    Lemma~\ref{lemma:two-bounds}, this implies equality.\qed
\end{proof}

\begin{theorem}\label{thm:two-exact}
Let $B=\{z_1,z_2\}$ with $z_1 < z_2$ and $x=\max A$, and let $t^* = \mathit{TD}(A,\{z_2\})$. Then,
\[
\mathit{TD}(A,\{z_1,z_2\}) =
\begin{cases}
1, & \text{if } \exists\,a,b\in A:\ a+b=z_1+z_2,\\[1mm]
2, & \text{if } t^* = 1 \text{ and no such } a,b \text{ exist},\\[1mm]
t^*, & \text{if the condition of Lemma~\ref{lemma:feasibility} holds},\\[1mm]
t^*+1, & \text{otherwise.}
\end{cases}
\]

Moreover, $\mathit{TD}(A,\{z_1,z_2\})$ can be computed in
$O(|A| + \log_2 \tfrac{z_2}{x})$ time.
\end{theorem}

\begin{proof}
If there exist $a,b\in A$ with $a+b=z_1+z_2$, Lemma~\ref{lemma:interval} shows that one
step suffices. If no such pair exists and $z_2 < x$, we know from Theorem~\ref{thm:b1} and
Theorem~\ref{thm:additive-Bminus1} that $t^*=1$ and $\mathit{TD}(A,\{z_1,z_2\})\in\{1,2\}$.
Since the value cannot be $1$, it must be $2$. The same reasoning
applies when $z_2 > x$ but $t^*=1$ (i.e., $z_2\in(x,2x]$): the lower bound $1$ is ruled out
by the first test, so the value is $2$. In the remaining case $z_2 > x$ and $t^*\ge 2$, Lemma~\ref{lemma:feasibility} shows that
$\mathit{TD}(A,\{z_1,z_2\}) = t^*$ if and only if either $z_1$ can be obtained during the doubling construction for $z_2$ or both $z_1$ and $z_2$ can be produced in the last step.
Otherwise, the distance must be $t^*+1$ by Lemma~\ref{lemma:two-bounds}. The running time is dominated by computing $x=\max A$ and evaluating in loop the feasibility of obtaining $z_1$ during the doubling construction for $z_2$ (with constant-time checks per iteration), which yields
$O(|A| + t^*) = O(|A| + \log_2 \tfrac{z_2}{x})$. Checking
whether there exist $a,b\in A$ with $a+b=z_1+z_2$ can be implemented in
$O(|A|)$ time by scanning $A$ once and storing its elements in a hash
table. \qed
\end{proof}

\section{Exact Algorithm for Constant $|B|$}
\label{sec:exact_constant}

In this section we give an exact pseudo-polynomial time algorithm for the unary translocation distance when $|B|$ is constant. The algorithm is based on a last-step case analysis and uses the exact algorithms for the cases $|B|=1$ and $|B|=2$ as subroutines.

Let $B=\{z_1,z_2,\dots,z_c\}$ with $z_1<z_2<\dots<z_c$, where $c=|B|$, and let $x=\max A$. We assume throughout this section that $A\cap B=\emptyset$. Let $z_{\max}=z_c$ and let $t^*=\mathit{TD}(A,\{z_{\max}\})$.
By the additive approximation result of Theorem~\ref{thm:additive-Bminus1}, we have $t^*\leq \mathit{TD}(A,B)\leq t^*+c-1$. On the other hand, one translocation can produce at most two new target values, so $\mathit{TD}(A,B)\geq \lceil \frac{c}{2}\rceil$. Therefore, $\max\{t^*,\lceil \frac{c}{2}\rceil\}\leq \mathit{TD}(A,B)\leq t^*+c-1$.

Thus, it is sufficient to test whether the target set can be produced in $K$ steps, for values of $K$ in the interval $\{\max\{t^*,\lceil \frac{c}{2}\rceil\},\dots,t^*+c-1\}$. Since a sequence using at most $K$ steps is also valid for any larger step bound, a YES answer for $K$ remains YES for all larger values. Therefore, we can use binary search on this interval.
%
%
The main difficulty, compared to the case $|B|=2$, is that a target may be produced together with an auxiliary value which is later used as a helper. To handle all such possibilities, we use a recursive last-step test.

For a finite set $C\subseteq \mathbb{N}$ with $|C|\leq c$ and an integer $K\geq 0$, let $D(C,K)$ denote the decision problem asking whether all values in $C$ can be produced in at most $K$ unary translocations from $A$.
If $C=\emptyset$, then $D(C,K)$ is true. If $|C|=1$, we use the exact singleton formula from Theorem~\ref{thm:b1}. If $|C|=2$, we use the exact algorithm from Theorem~\ref{thm:two-exact}. Therefore, the only remaining case is $|C|\geq 3$.

Suppose $|C|\geq 3$. If $K=0$, then $D(C,K)$ is false. Otherwise, let $U_{K-1}=x\cdot 2^{K-1}$. By Lemma~\ref{lemma:doubling}, every value available before the last step of a $K$-step sequence is at most $U_{K-1}$. In any minimal sequence, the last step must produce at least one value from $C$. Otherwise, it could be removed. Since one translocation has two outputs, the last useful step produces either exactly one value from $C$, or exactly two values from $C$. We distinguish the following two cases.

\paragraph{Case 1: the last step produces exactly one required value.}
Let this value be $r\in C$. Suppose the inputs of the last step are $p$ and $q$. Since one output is $r$, we have $p+q\geq r$. Let $m=\max\{p,q\}$. Then $m$ is available before the last step and $2m\geq r$. Therefore, this case is feasible if and only if there exist $r\in C$ and $m\in\{\lceil \frac{r}{2}\rceil,\dots,U_{K-1}\}$ such that $D((C\setminus\{r\})\cup\{m\},K-1)$ is true.

\paragraph{Case 2: the last step produces exactly two required values.}
Let these values be $r,s\in C$, $r\neq s$. If the inputs of the last step are $p$ and $q$, then $p+q=r+s$. Thus, before the last step, the prefix has to produce all values in $C\setminus\{r,s\}$ together with the two helper values $p$ and $q$. Hence this case is feasible if and only if there exist distinct $r,s\in C$ and values $p,q\in\{0,\dots,U_{K-1}\}$ such that $p+q=r+s$ and $D((C\setminus\{r,s\})\cup\{p,q\},K-1)$ is true.

The resulting decision procedure is described in Algorithm~\ref{alg:decide_constant}. The procedure is memoized: each pair $(C,K)$ is evaluated at most once, with $C$ stored in sorted order after removing the elements of $A$.

\begin{algorithm}[t]
\caption{\textsc{D}$(C,K)$}
\label{alg:decide_constant}
\begin{algorithmic}[1]
\If{$C=\emptyset$}
    \State \Return YES
\EndIf
\If{$K<0$}
    \State \Return NO
\EndIf
\If{$|C|=1$}
    \State \Return whether the exact algorithm for $|B| = 1$ gives distance at most $K$
\EndIf
\If{$|C|=2$}
    \State \Return whether the exact algorithm for $|B|=2$ gives distance at most $K$
\EndIf
\If{$K=0$}
    \State \Return NO
\EndIf
\State Let $U\gets x\cdot 2^{K-1}$
\ForAll{$r\in C$}
    \For{$m=\lceil \frac{r}{2}\rceil,\dots,U$}
        \If{\textsc{D}$((C\setminus\{r\})\cup\{m\},K-1)$ = YES}
            \State \Return YES
        \EndIf
    \EndFor
\EndFor
\ForAll{unordered pairs $\{r,s\}\subseteq C$}
    \For{$p=0,\dots,r+s$}
        \State $q\gets r+s-p$
        \If{$p\leq U$ and $q\leq U$}
            \If{\textsc{D}$((C\setminus\{r,s\})\cup\{p,q\},K-1)$ = YES}
                \State \Return YES
            \EndIf
        \EndIf
    \EndFor
\EndFor
\State \Return NO
\end{algorithmic}
\end{algorithm}

The exact algorithm for constant $|B|$ is given in Algorithm~\ref{alg:exact_constant}. It performs binary search on the interval $\{\max\{t^*,\lceil \frac{c}{2}\rceil\},\dots,t^*+c-1\}$.

\begin{algorithm}[t]
\caption{}
\label{alg:exact_constant}
\begin{algorithmic}[1]
\State $c\gets |B|$
\State Sort the targets so that $B=\{z_1,z_2,\dots,z_c\}$ and $z_1<z_2<\dots<z_c$
\State $x\gets \max A$
\State $t^*\gets \mathit{TD}(A,\{z_c\})$
\State $L\gets \max\{t^*,\lceil \frac{c}{2}\rceil\}$
\State $R\gets t^*+c-1$
\State $SOL\gets 0$
\While{$L \leq R$}
    \State $K\gets \lfloor \frac{L +R}{2}\rfloor$
    \If{\textsc{D}$(B,K)$ = YES}
        \State $R\gets K - 1$
        \State $SOL\gets K$
    \Else
        \State $L\gets K+1$
    \EndIf
\EndWhile
\State \Return $SOL$
\end{algorithmic}
\end{algorithm}

\begin{theorem}\label{thm:constant-exact}
Algorithm~\ref{alg:exact_constant} computes $\mathit{TD}(A,B)$ exactly in time $O(|A|^2+K_{\max}c^2U^{c+1})$,
where $c=|B|$, $t^*=\mathit{TD}(A,\{\max B\})$, $K_{\max}=t^*+c-1$, $x=\max A$, and $U=x2^{K_{\max}}$.
\end{theorem}

\begin{proof}
We first prove the correctness of the decision procedure $\textsc{D}(C,K)$. The proof is by induction on $K$. If $C=\emptyset$, then the answer is YES. If $|C|=1$, correctness follows from the exact singleton formula in Theorem~\ref{thm:b1}. If $|C|=2$, correctness follows from Theorem~\ref{thm:two-exact}. It remains to consider the case $|C|\geq 3$.

Assume first that $\textsc{D}(C,K)$ should return YES, and consider a feasible sequence of length at most $K$ producing all values in $C$. Since one translocation has two outputs, the last operation produces either exactly one value from $C$ or exactly two values from $C$.

Suppose first that the last operation produces exactly one value $r\in C$. Let $p$ and $q$ be the two inputs of this last operation, and let $m=\max\{p,q\}$. Since one output is $r$, we have $p+q\geq r$, and therefore $2m\geq r$. The value $m$ is available before the last step. Once $m$ is available it may be reused, so the last step can be replaced by the translocation $(m,m)\vdash(r,2m-r)$. Hence the prefix of the sequence produces all values in $(C\setminus\{r\})\cup\{m\}$ in at most $K-1$ steps. This is exactly Case~1 of the decision procedure.

Suppose now that the last operation produces exactly two values $r,s\in C$, with $r\neq s$. If its inputs are $p$ and $q$, then the sum-preservation property gives $p+q=r+s$. Hence the prefix of the sequence produces all values in $(C\setminus\{r,s\})\cup\{p,q\}$ in at most $K-1$ steps. This is exactly Case~2 of the decision procedure.

Conversely, suppose that one of the recursive tests used by $\textsc{D}(C,K)$ returns YES. In Case~1, the induction hypothesis gives a sequence of length at most $K-1$ producing all values in $(C\setminus\{r\})\cup\{m\}$, and since $2m\geq r$, one additional translocation from $(m,m)$ produces $r$. Thus all values in $C$ are produced in at most $K$ steps. In Case~2, the induction hypothesis gives a sequence of length at most $K-1$ producing all values in $(C\setminus\{r,s\})\cup\{p,q\}$, and since $p+q=r+s$, one additional translocation from $(p,q)$ produces $(r,s)$. Again, all values in $C$ are produced in at most $K$ steps. Therefore, $\textsc{D}(C,K)$ is correct.

We now prove the correctness of Algorithm~\ref{alg:exact_constant}. Let $z_c=\max B$ and $t^*=\mathit{TD}(A,\{z_c\})$. Any sequence producing $B$ must produce $z_c$, so $\mathit{TD}(A,B)\geq t^*$. Also, since one translocation has only two outputs and $A\cap B=\emptyset$, we have $\mathit{TD}(A,B)\geq \lceil \frac{c}{2}\rceil$. Hence $\mathit{TD}(A,B)\geq L$, where $L=\max\{t^*,\lceil \frac{c}{2}\rceil\}$.

On the other hand, by Theorem~\ref{thm:additive-Bminus1}, we have $\mathit{TD}(A,B)\leq t^*+c-1=R$. Thus the optimum belongs to the interval $\{L,L+1,\dots,R\}$.

The function $\textsc{D}(B,K)$ is monotone in $K$: if $B$ can be produced in at most $K$ translocations, then it can also be produced in at most $K'$ translocations for every $K'\geq K$. Therefore, binary search over the interval $\{L,L+1,\dots,R\}$ returns the smallest value $K$ for which $\textsc{D}(B,K)$ is true. By the correctness of $\textsc{D}$, this value is exactly $\mathit{TD}(A,B)$.

We analyze the running time of Algorithm~\ref{alg:exact_constant}. We store $A$ in a hash table, so membership queries in $A$ take expected constant time. We also precompute pair sums of elements of $A$ in $O(|A|^2)$ time, so the exact algorithms for $|B|=1$ and $|B|=2$ can answer the required pair-sum queries in constant time after preprocessing.

Let $K_{\max}=t^*+c-1$ and let $U=x2^{K_{\max}}$. By Lemma~\ref{lemma:doubling}, all values considered by the recursive procedure lie in $\{0,\dots,U\}$.

The procedure $\textsc{D}$ is memoized. A memoized state is a pair $(C,K)$, where $K\leq K_{\max}$ and $C$ is a set of at most $c$ values from $\{0,\dots,U\}$. Hence the number of possible states is at most $O(K_{\max}U^c)$, for fixed $c$.

For a state with $|C|\geq 3$, Case~1 enumerates at most $O(cU)$ choices, because we choose $r\in C$ and then enumerate $m$. Case~2 enumerates at most $O(c^2U)$ choices, because we choose an unordered pair $\{r,s\}\subseteq C$ and then enumerate $p$, while $q$ is determined by $q=r+s-p$. Thus each state creates at most $O(c^2U)$ recursive tests.

Therefore the total running time is $O(|A|^2+K_{\max}c^2U^{c+1})$.

If $z_c>x$, then $t^*=\lceil\log_2(z_c/x)\rceil$ and $x2^{t^*-1}<z_c\leq x2^{t^*}$. Since $K_{\max}=t^*+c-1$, we have $U=x2^{K_{\max}}=x2^{t^*+c-1}<2^c z_c$. Hence in this case the running time is $O(|A|^2+(t^*+c)c^2(2^c z_c)^{c+1})$.

If $z_c<x$, then $t^*=1$ and $K_{\max}=c$. Hence $U=x2^c$, and the running time is $O(|A|^2+c^3(x2^c)^{c+1})$.

\qed
\end{proof}

\begin{corollary}
The unary translocation distance problem is fixed-parameter tractable with
respect to the combined parameter $(|B|,\max\{\max A,\max B\})$.
\end{corollary}

\section{ILP Formulation}
\label{sec:ilp_noncontiguous}

In this section we give an integer linear programming formulation for computing the unary translocation distance. Starting from an upper bound $H$ on the number of operations, obtained for instance from our $2$-approximation algorithm, we build a time-expanded model over $H$ steps. The ILP simultaneously selects which translocation rule is applied at each step and tracks which lengths become available. The objective value then coincides with the exact unary translocation distance.

\begin{observation}\label{obs:split_sum}
Let $x,y,u,v\in\mathbb{N}$. There exists a unary translocation $(x,y)\vdash (u,v)$ if and only if $x+y=u+v$.
\end{observation}

Given an instance $(A,B)$, let $x=\max A$. By Lemma~\ref{lemma:doubling}, no length larger than $x\cdot 2^H$ can be obtained within $H$ translocations. Therefore, we take $L=\{0,\dots,x\cdot 2^H\}$. Let $\mathcal{O}$ be the set of all admissible unary translocation rules over $L$. Each operation $o\in\mathcal{O}$ is a quadruple $o=(x_o,y_o,x'_o,y'_o)$ such that $x_o,y_o,x'_o,y'_o\in L$ and $x_o+y_o=x'_o+y'_o$. By Observation~\ref{obs:split_sum}, this condition is equivalent to the existence of a unary translocation $(x_o,y_o)\vdash(x'_o,y'_o)$.

We introduce binary variables $w_{\ell,t}$ indicating whether length $\ell$ is available after at most $t$ operations, and binary variables $u_{o,t}$ indicating whether operation $o$ is applied at step $t$. We also use binary variables $z_t$ for step usage and an integer variable $T$ for the number of operations. The exact distance is the optimum value of the following ILP.

\begin{align}
  \min\ & T \label{obj:translocation_T} \\[0.5em]
  \text{s.t.}\quad
  & w_{\ell,0} = 
    \begin{cases}
      1 & \text{if } \ell \in A,\\
      0 & \text{if } \ell \notin A,
    \end{cases}
    && \forall \ell \in L, \label{cons:init}\\[0.7em]
  & w_{\ell,t} \ge w_{\ell,t-1} 
    && \forall \ell \in L,\ \forall t = 1,\dots,H, \label{cons:mono}\\[0.5em]
  & u_{o,t} \le w_{x_o,t-1}, \quad
    u_{o,t} \le w_{y_o,t-1}
    && \forall o \in \mathcal{O},\ \forall t = 1,\dots,H, \label{cons:inputs}\\[0.5em]
  & w_{x'_o,t} \ge u_{o,t}, \quad
    w_{y'_o,t} \ge u_{o,t}
    && \forall o \in \mathcal{O},\ \forall t = 1,\dots,H, \label{cons:outputs}\\[0.5em]
  & w_{\ell,t} \le w_{\ell,t-1}
    + \displaystyle\sum_{o\in\mathcal{O}:\ \ell\in\{x'_o,y'_o\}} u_{o,t}
    && \forall \ell\in L,\ \forall t=1,\dots,H, \label{cons:no_free}\\[0.5em]
  & \displaystyle\sum_{o \in \mathcal{O}} u_{o,t} \le 1
    && \forall t = 1,\dots,H, \label{cons:oneop}\\[0.5em]
  & \displaystyle\sum_{o \in \mathcal{O}} u_{o,t} \le z_t,
    \quad
    z_t \le \displaystyle\sum_{o \in \mathcal{O}} u_{o,t}
    && \forall t = 1,\dots,H, \label{cons:z_link}\\[0.5em]
  & z_t \ge z_{t+1}
    && \forall t = 1,\dots,H-1, \label{cons:prefix}\\[0.5em]
  & T = \displaystyle\sum_{t=1}^H z_t, \label{cons:T_def}\\[0.5em]
  & w_{\ell,H} = 1
    && \forall \ell \in B, \label{cons:targets}\\[0.5em]
  & w_{\ell,t} \in \{0,1\}
    && \forall \ell \in L,\ \forall t = 0,\dots,H, \label{cons:dom_w}\\
  & u_{o,t} \in \{0,1\}
    && \forall o \in \mathcal{O},\ \forall t = 1,\dots,H, \label{cons:dom_u}\\
  & z_t \in \{0,1\}
    && \forall t = 1,\dots,H, \label{cons:dom_z}\\
  & T \in \mathbb{Z},\quad 0 \le T \le H.
    && \label{cons:dom_T}
\end{align}

Constraint~\eqref{cons:init} initializes availability to the input set $A$, and constraint~\eqref{cons:mono} enforces that lengths remain available once they have been produced. Constraints~\eqref{cons:inputs} guarantee that an operation can be used at step $t$ only if both its input lengths are already available at time $t-1$. Constraints~\eqref{cons:outputs} state that if an operation is used, then its output lengths become available at time $t$.

The additional constraint~\eqref{cons:no_free} is essential for correctness. It prevents a length from becoming available unless it was already available at the previous time step or is produced by the operation selected at the current step. Without this constraint, the variables $w_{\ell,t}$ could be set to $1$ without any operation actually producing $\ell$.

Constraint~\eqref{cons:oneop} enforces a sequential model with at most one operation per step. Constraints~\eqref{cons:z_link} define $z_t$ as the indicator of whether step $t$ is used. Constraint~\eqref{cons:prefix} removes symmetries by enforcing that used steps form a prefix. Constraint~\eqref{cons:T_def} defines $T$ as the number of used operations. Finally, constraint~\eqref{cons:targets} requires that all target lengths are available by time $H$.

Since $H$ is chosen as the length of a feasible solution, the optimal value of the ILP equals the exact unary translocation distance.

Let $n=|L|$. In the explicit quadruple formulation, an operation $o=(x_o,y_o,x'_o,y'_o)$ is determined by three choices: $x_o$, $y_o$, and $x'_o$, since the fourth value is forced by $y'_o=x_o+y_o-x'_o$. Hence $|\mathcal{O}|=O(n^3)$.

The number of variables is as follows. There are $(H+1)n$ availability variables $w_{\ell,t}$, $H|\mathcal{O}|$ operation-selection variables $u_{o,t}$, $H$ step-usage variables $z_t$, and one variable $T$. Therefore, the number of variables is $O(Hn+H|\mathcal{O}|)=O(Hn^3)$.

The number of constraints is also dominated by the operation-selection constraints. Constraints~\eqref{cons:init}, \eqref{cons:mono}, \eqref{cons:no_free}, and \eqref{cons:targets} contribute $O(Hn+|B|)$ constraints. Constraints~\eqref{cons:inputs} and~\eqref{cons:outputs} contribute $O(H|\mathcal{O}|)$ constraints. The remaining constraints contribute only $O(H)$ constraints. Therefore, the number of constraints is $O(Hn+H|\mathcal{O}|)=O(Hn^3)$.

Thus, although the model is exact, it is pseudo-polynomial and can become large even for small values of $H$. 

\subsection{A more compact pair-sum formulation}

The explicit formulation above enumerates complete quadruples $(x,y,x',y')$. This is not necessary. Since a unary translocation is completely characterized by preserving the sum, we may instead choose an input pair and an output pair with the same sum.

For each integer $s$, let $\mathcal{P}_s$ be the set of unordered pairs of lengths in $L$ with sum $s$, namely $\mathcal{P}_s=\{\{a,b\}\mid a,b\in L,\ a\le b,\ a+b=s\}$. Let $\mathcal{P}=\displaystyle\bigcup_s\mathcal{P}_s$. Then $|\mathcal{P}|=O(n^2)$.

In the compact formulation, for every step $t$ we introduce binary variables $a_{p,t}$ indicating that pair $p$ is chosen as the input pair, binary variables $b_{p,t}$ indicating that pair $p$ is chosen as the output pair, and binary variables $q_{s,t}$ indicating that the common sum at step $t$ is $s$. The constraints are:

\begin{align}
  & \displaystyle\sum_s q_{s,t}=z_t
    && \forall t=1,\dots,H, \label{cons:compact_sum}\\
  & \displaystyle\sum_{p\in\mathcal{P}_s} a_{p,t}=q_{s,t}
    && \forall s,\ \forall t=1,\dots,H, \label{cons:compact_input}\\
  & \displaystyle\sum_{p\in\mathcal{P}_s} b_{p,t}=q_{s,t}
    && \forall s,\ \forall t=1,\dots,H. \label{cons:compact_output}
\end{align}

These constraints ensure that, if step $t$ is used, then exactly one input pair and one output pair are chosen, and both pairs have the same sum. The availability and output constraints are then written with respect to the chosen input and output pairs. For example, if $\ell$ belongs to the input pair $p$, then we impose $a_{p,t}\le w_{\ell,t-1}$. If $\ell$ belongs to the output pair $p$, then we impose $w_{\ell,t}\ge b_{p,t}$. The no-free-availability constraint becomes $w_{\ell,t}\le w_{\ell,t-1}+\displaystyle\sum_{p\in\mathcal{P}:\ \ell\in p} b_{p,t}$.

This formulation is equivalent to the explicit quadruple formulation, because every input pair and output pair with the same sum define a valid unary translocation. However, it uses only $O(Hn^2)$ pair-selection variables instead of $O(Hn^3)$ operation-selection variables.

\subsection{Integrality gap of the LP relaxation}

We next observe that the LP relaxation of the formulation is not integral. The LP relaxation is obtained by replacing the constraints $w_{\ell,t},u_{o,t},z_t\in\{0,1\}$ with $w_{\ell,t},u_{o,t},z_t\in[0,1]$ and allowing $T$ to be continuous.

\begin{lemma}
The LP relaxation of the ILP has integrality gap at least $\frac{4}{3}$.
\end{lemma}

\begin{proof}
Consider the instance $A=\{1,29,2,38,3,47\}$ and $B=\{10,20,30\}$. The following three operations are admissible:
$(1,29)\vdash(10,20)$, $(2,38)\vdash(10,30)$, and $(3,47)\vdash(20,30)$.

The integer optimum is equal to $2$. Indeed, since $A\cap B=\emptyset$ and $|B|=3$, every feasible solution needs at least $\lceil \frac{3}{2}\rceil=2$ translocations. On the other hand, two translocations are sufficient, for example $(1,29)\vdash(10,20)$ and $(2,38)\vdash(10,30)$.

We present a fractional LP solution of value $\frac{3}{2}$. At the first time layer, take the first two operations with value $\frac{1}{2}$ each. Thus $z_1=1$. At the second time layer, take the third operation with value $\frac{1}{2}$, so $z_2=\frac{1}{2}$. The fractional production of the targets is sufficient to set $w_{10,2}=w_{20,2}=w_{30,2}=1$, while satisfying the no-free-availability constraints. Therefore the LP relaxation has value at most $\frac{3}{2}$.

Consequently, for this instance, the ratio between the integer optimum and the LP optimum is at least $\frac{2}{\frac{3}{2}}=\frac{4}{3}$. Hence the integrality gap of the relaxation is at least $\frac{4}{3}$.

\qed
\end{proof}

\section{Experiments}\label{sec:experiments}

In this section we present our experimental results. Each test instance consists of two finite disjoint sets of chromosome lengths
$A,B \subset \mathbb{N}$, representing the input and target genomes,
respectively. We generate several families of instances to stress the 2-approximation algorithm and different heuristics.
\par
\textit{Small instances (SMALL).}
This class is used to compare heuristics against more expensive methods (beam
search, simulated annealing). We first
choose integers $n,m$ with $1 \le n,m$ and $n+m \le 50$. We then draw
$n+m$ distinct integers uniformly at random from $\{1,\dots,50\}$, permute
them, and assign the first $n$ values to $A$ and the remaining $m$ values
to $B$. 
\par
\textit{Random baseline (RAND).}
These instances mimic unstructured data at different scales. We first choose a
scale parameter (small/medium/large), which determines the ranges for $|A|$,
$|B|$ and for the maximum length. For example, in the small scale we take
$1 \le |A|,|B| \le 10$ and lengths up to $10^3$, while in the large scale
we allow up to 500 values and lengths up to $10^6$.
For each instance we sample a value $x$ uniformly from a scale-dependent interval and generate $A$ as a random set of distinct integers in $[1,x]$ with $\max A=x$.  Generate $B$ as a random set of distinct integers in $[x+1,x+M]$, for a scale-dependent $M$, so that all elements of $B$ are strictly larger than $\max A$.
\par
\textit{Structured instances (STR).}
To test behaviour on more regular inputs we generate $A$ and $B$ from
different parametric sequence families. We fix a global range $[1,L]$ for
$A$ (with $L=10^5$ in the implementation) and independently choose the
size of each set in the interval $[5,100]$. For $A$ we pick one of the
following families at random:
\begin{enumerate}
  \item \emph{UNIFORM:} distinct values chosen uniformly at random in
        $[1,L]$;
  \item \emph{ARITHMETIC:} an arithmetic progression $a, a+d, \dots$ clipped
        to $[1,L]$;
  \item \emph{GEOMETRIC:} a geometric progression $a, ar, ar^2, \dots$ with
        integer ratio $r\in\{2,\dots,5\}$, scaled to fit in $[1,L]$;
  \item \emph{FIBONACCI:} a scaled and shifted prefix of the Fibonacci
        sequence, again clipped to $[1,L]$.
\end{enumerate}
For $B$ we select a \emph{different} family and generate lengths in an
interval $[M,M+L]$ that starts just above $\max A$.
%
\par
\textit{Two-approx–friendly instances (TA).}
This family is tailored to favour the 2-approximation algorithm.
We first generate $A$ as a random set with $\max A=x$ as in the random
baseline. Then we construct $B$ as an increasing sequence that roughly follows
a noisy doubling pattern: the $i$-th target is drawn uniformly from a small
interval around $x\cdot 2^{i+1}$, truncated to stay in $[2x+1,10^9]$.
This ensures that target lengths grow quickly and are always much larger than
$\max A$, making the doubling strategy close to optimal in many cases and
guaranteeing $A\cap B=\varnothing$.
\par
\textit{Upper-bound–branch instances (UB).}
These instances are designed so that the simple analytical upper bound triggers
its special case where $\max A > \max(B)$. We first generate $A$ as in the
random baseline with maximum value $x$. We then generate $B$ as a random
set of distinct integers in $[1,x-1]$ that is disjoint from $A$. This
guarantees $\max A > \max(B)$ and $A\cap B=\varnothing$, activating the
branch of the upper-bound formula that depends only on $\max A$ and $|B|$.
\par
\textit{Slow-increasing targets (SLOW).}
This class is mildly adversarial for doubling-based heuristics. We first
generate $A$ as in the random baseline with maximum value $x$. We then
construct $B$ as a slowly increasing sequence starting at $x+1$:
$
b_1 = x+1,\quad
b_{i+1} = b_i + \delta_i,\quad \delta_i \in \{1,\dots,5\}.
$
All targets are therefore just above $\max A$ and remain relatively close
to each other, which tends to reduce the effectiveness of pure doubling
strategies.
\par
\textit{Addition-chain–based instances (CHAIN).}
To connect with classical addition-chain constructions we generate instances
from random addition chains. We first choose an integer $N$ (uniformly
in $[50,10^5]$) and compute a random addition chain
$1 = c_1 < c_2 < \dots < c_\ell = N$ by repeatedly adding two existing
elements $c_i + c_j$ chosen uniformly among all sums that remain $\le N$.
We then consider the multiset of differences
$
p_1 = c_1,\quad p_2 = c_2 - c_1,\quad \dots,\quad p_\ell = c_\ell - c_{\ell-1},
$
remove duplicates, and obtain a set $\{p_1,\dots,p_D\}$. Finally we pick a
random split index $t\in\{1,\dots,D-1\}$ and define
$
A = \{p_1,\dots,p_t\}, B = \{p_{t+1},\dots,p_D\}.
$


\paragraph{Evaluation on small instances.}
For the \texttt{SMALL} class we evaluate three algorithms:
\begin{enumerate}
  \item 2-approximation.
  \item Simulated annealing (SA): a stochastic local search over length-$H$ sequences of abstract translocations, where $H$ is the 2-approximation value. The energy heavily penalizes uncovered targets and then minimizes the first hit. We run up to $5\cdot 10^8$ iterations with geometric cooling ($T_0=10$, $T_{\min}=10^{-3}$, $\alpha=0.995$) and report the best solution. If SA does not cover all of $B$ within $H$, we output $-1$.
  \item Beam search: states are reachable sets $S$ under non-contiguous closure, expanded breadth-wise up to depth $H$, keeping the best $K=50$ states by a score combining missing targets and their distance to $S$. If a state covering $B$ is found within $H$, we report its depth; otherwise we output $-1$.
\end{enumerate}

Over the 20 \texttt{SMALL} instances, beam search finds a feasible solution
within depth $H$ in 8 cases (and fails in the remaining 12 cases, reported as $-1$).
On the 8 successful instances, beam search improves on the 2-approximation in
4 cases and matches it in the other 4, with an average gain of 2 translocations
and a maximum gain of 3.
Simulated annealing is more robust: it succeeds on 16 instances (4 failures),
improving on the 2-approximation in 9 cases and matching it in the other 7,
with an average gain of about 1.4 and a maximum gain of 3.
Conditioned on the instances where each heuristic succeeds, the mean distance
decreases from 4.1 to 3.1 for beam search and from 12.6 to 11.8 for SA.
Overall, SA has a higher success rate, while beam search can yield larger
improvements when it succeeds.
\par
\textit{Evaluation on larger instances.}
For the larger classes (RAND, STR, TA, UB, SLOW, CHAIN) the search space grows
quickly with the number and magnitude of the lengths, and SA or beam search do
not consistently improve over the 2-approximation while incurring a
substantial running-time cost. We therefore restrict the comparison to the
2-approximation and the lower/upper bounds of
Theorem~\ref{thm:additive-Bminus1}. Table~\ref{tab:large-results} summarises,
for each class, the average distance returned by the 2-approximation together
with its average gap to the lower and upper bounds. Across all families the lower bound $t$ is rarely informative: the average
gap $\text{2approx} - t$ is large (between roughly 6 and 130 steps per
class, and about $42.6$ on average over all 120 tests), and it is never
tight on our instances. In contrast, the additive upper bound
$t + |B| - 1$ is very accurate in practice: the average slack
$(t + |B| - 1) - \text{2approx}$ is at most a few steps in every class
(about $1.8$ on average over all tests), and the bound is exactly attained
by the 2-approximation on 54 out of 120 instances (45\%), including all
UB and SLOW instances. This shows that, although the lower bound is mainly of
theoretical interest, the additive $|B|-1$ upper bound captures very closely the behaviour of the 2-approximation on larger inputs.

\begin{figure}[t]
  \centering
  \begin{minipage}{0.48\textwidth}
    \centering
    \includegraphics[width=\linewidth]{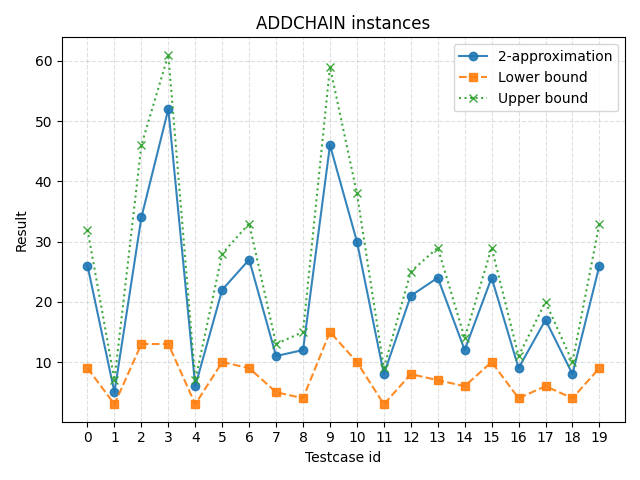}
  \end{minipage}\hfill
  \begin{minipage}{0.48\textwidth}
    \centering
    \includegraphics[width=\linewidth]{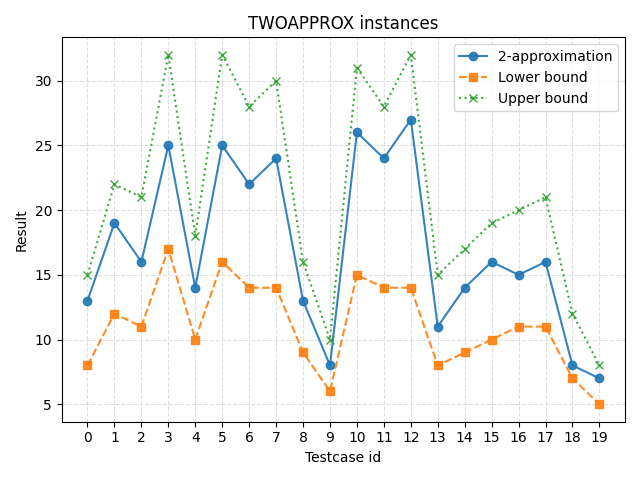}
  \end{minipage}
  \caption{Values of the 2-approximation, the lower bound $t$, and the upper
  bound $t + |B| - 1$ for the ADDCHAIN (left) and TWOAPPROX (right) instances.}
  \label{fig:large-instances}
\end{figure}

\begin{table}[t]
\centering
\begin{tabular}{lrrrr}
\toprule
Class & mean(2-approx) & mean gap$_\text{LB}$ & mean gap$_\text{UB}$ & \#tight UB \\
\midrule
RAND   & 131.8 & 129.6 & 1.15 & 1 \\
STR    &  60.5 &  59.1 & 0.40 & 13 \\
TA     &  17.2 &   6.1 & 4.20 & 0 \\
UB     &  29.4 &  28.4 & 0.00 & 20 \\
SLOW   &  20.1 &  19.1 & 0.00 & 20 \\
CHAIN  &  21.0 &  13.5 & 5.00 & 0 \\
\midrule
All (avg) &  -- &  42.6 & 1.80 & 54 / 120 \\
\bottomrule
\end{tabular}
\caption{Average value of the 2-approximation and its average gaps to the
lower bound $t$ and the upper bound $t + |B| - 1$ from
Theorem~\ref{thm:additive-Bminus1} on the larger instance classes.}
\label{tab:large-results}
\vspace*{-0.7cm}
\end{table}

Figure~\ref{fig:large-instances} details the behaviour of the 2-approximation on the \texttt{CHAIN} and \texttt{TA} families. In both plots, the 2-approximation (solid blue line) tracks the additive upper bound $t + |B| - 1$ (dotted green line) much more closely than the lower bound $t$ (dashed orange line), which remains significantly detached from the true solution cost. For the \texttt{CHAIN} instances, the correlation is particularly striking: while not strictly tight (mean gap $\approx 5.0$), the algorithm's performance mirrors the upper bound's fluctuations almost perfectly, capturing every spike in difficulty. The \texttt{TA} instances display a similar qualitative pattern, where the 2-approximation consistently operates near the upper limit despite a slightly more variable gap. This confirms that for these larger inputs, the structural term $t + |B| - 1$ is a far superior predictor of complexity than the theoretical lower bound $t$. All code, input instances, and detailed numerical results for this section are
publicly available in our GitHub repository~\citep{git}.

\newcommand{\range}[2]{\{#1,\ldots,#2\}}

\section{Parametrized Algorithm in Translocation Distance and $max A$}
\label{sec:fpt_noncontiguous}

In this section, we give an FPT algorithm for the decision version of the unary translocation distance problem, parameterized by the solution value and by $\max A$. Given finite sets $A,B\subseteq \mathbb{N}$ such that $A \cap B = \emptyset$ and an integer $k$, the decision problem asks whether $\mathit{TD}(A,B)\leq k$.

Let $x=\max A$. By Lemma~\ref{lemma:doubling}, every length that can be obtained after at most $k$ translocations is at most $x\cdot 2^k$. Hence, for the decision problem with budget $k$, it is sufficient to consider the finite set of lengths $L_k=\{0,1,\dots,x\cdot 2^k\}$.

Since the elements of $A$ are available from the beginning and can be reused arbitrarily many times, a state only needs to store the set of lengths outside $A$ that have already been produced.

\begin{algorithm}[t]
\caption{\textsc{FPT}$(A,B,k)$}
\label{alg:noncontiguous_fpt}
\begin{algorithmic}[1]
\Require Finite sets $A,B\subseteq \mathbb{N}$ and an integer $k$
\Ensure YES iff $\mathit{TD}(A,B)\leq k$
\State $x\gets \max A$
\State $L\gets \{0,1,\dots,x\cdot 2^k\}$
\If{there exists $b\in B'$ such that $b\notin L$}
    \State \Return NO
\EndIf
\State $\mathcal{R}_0\gets \{\emptyset\}$
\For{$t=0,1,\dots,k$}
    \If{there exists $C\in \mathcal{R}_t$ such that $B'\subseteq C$}
        \State \Return YES
    \EndIf
    \If{$t=k$}
        \State \textbf{continue}
    \EndIf
    \State $\mathcal{R}_{t+1}\gets \emptyset$
    \ForAll{$C\in \mathcal{R}_t$}
        \ForAll{$p\in A\cup C$}
            \ForAll{$q\in A\cup C$}
                \For{$u=0,1,\dots,p+q$}
                    \State $v\gets p+q-u$
                    \If{$u\in L$ and $v\in L$}
                        \State $C'\gets C\cup(\{u,v\}\setminus A)$
                        \State add $C'$ to $\mathcal{R}_{t+1}$
                    \EndIf
                \EndFor
            \EndFor
        \EndFor
    \EndFor
\EndFor
\State \Return NO
\end{algorithmic}
\end{algorithm}

\begin{theorem}\label{thm:noncontiguous_fpt}
Algorithm~\ref{alg:noncontiguous_fpt} decides whether $\mathit{TD}(A,B)\leq k$ in time $(2^k x)^{O(k)}\cdot |I|^{O(1)}$, where $x=\max A$ and $|I|$ denotes the input size. Hence, the unary translocation distance problem is fixed-parameter tractable with respect to the combined parameter $(k,x)$.
\end{theorem}

\begin{proof}
We first prove correctness. Since the elements of $A$ are available in arbitrarily many copies, they do not have to be stored in the state. A state $C$ represents exactly the set of lengths outside $A$ that have been produced so far.

The initial state is $\emptyset$, since no length outside $A$ has been produced before any operation. From a state $C$, Algorithm~\ref{alg:noncontiguous_fpt} considers all possible choices of two available input lengths $p,q\in A\cup C$ and all possible splits $u+v=p+q$. By the definition of unary translocation, this enumerates all possible one-step extensions of the current sequence. The new state is $C\cup(\{u,v\}\setminus A)$, because newly produced lengths outside $A$ become available forever, while lengths in $A$ are already available.

By Lemma~\ref{lemma:doubling}, no length larger than $x\cdot 2^k$ is needed in a sequence of length at most $k$. Hence restricting the search to $L_k=\{0,1,\dots,x\cdot 2^k\}$ loses no feasible solution. Thus the algorithm explores exactly all unary translocation sequences of length at most $k$, up to the set of produced lengths outside $A$. It returns YES precisely when some reachable state contains all targets in $B$, which is equivalent to $\mathit{TD}(A,B)\leq k$.

Let $n=|L_k|$. Since $L_k=\{0,1,\dots,x\cdot 2^k\}$, we have $n=x\cdot 2^k+1$. In at most $k$ translocations, at most $2k$ lengths outside $A$ can be produced, because each translocation produces at most two output lengths. Therefore every reachable state $C$ satisfies $|C|\leq 2k$. Moreover, every element stored in $C$ belongs to $L_k$. Hence the number of possible states is at most the number of subsets of $L_k$ of size at most $2k$, namely $\displaystyle\sum_{i=0}^{2k}\binom{n}{i}$. Using the bound $\binom{n}{i}\leq n^i$, we obtain $\displaystyle\sum_{i=0}^{2k}\binom{n}{i}\leq \displaystyle\sum_{i=0}^{2k}n^i\leq (2k+1)n^{2k}$.
Since $n=x\cdot 2^k+1$, this quantity is bounded by $(2^k x)^{O(k)}$.

We now bound the number of transitions generated from one state. From a state $C$, the algorithm chooses two input lengths from $A\cup C$. Since $|C|\leq 2k$, there are at most $(|A|+2k)^2$ ordered choices for the two inputs. Once the two inputs $p$ and $q$ are fixed, a unary translocation is determined by choosing one output $u$ with $0\leq u\leq p+q$, the other output being $v=p+q-u$. By Lemma~\ref{lemma:doubling}, all relevant lengths are at most $x\cdot 2^k$, and therefore it is enough to consider at most $O(x\cdot 2^k)$ possible splits. Thus the number of transitions generated from a single state is at most $O((|A|+2k)^2\cdot x\cdot 2^k)$. Thus, the running time is $(2^k x)^{O(k)}\cdot |I|^{O(1)}$.
\end{proof}

\section{Conclusion and Open Problems}

In this paper, we studied the unary translocation distance problem from both theoretical and computational perspectives. We proved that computing the unary translocation distance is strongly NP-hard, thereby settling an open problem posed in the ISCO 2026 conference version of this work. On the positive side, we gave an exact pseudo-polynomial algorithm for every fixed constant value of $|B|$, extending the previously known exact algorithms for $|B|\leq 2$. For arbitrary target sets, we presented a $2$-approximation algorithm, an additive $(|B|-1)$-approximation algorithm, and showed that the additive algorithm also gives a $3$-approximation. We also derived parameterized algorithms, including algorithms parameterized by the maximum value in the input set together with the optimum distance, and by the maximum value in the target set together with $|B|$. From an optimization point of view, we proposed an integer linear programming formulation that gives an exact model for the problem, analyzed its size, and showed that the LP relaxation has integrality gap at least $\frac{4}{3}$. Finally, our computational experiments compared the $2$-approximation algorithm with beam search and simulated annealing.

There are several directions for future research. A first open problem is whether the $2$-approximation ratio can be improved. Another natural question is whether the exact pseudo-polynomial algorithm for constant $|B|$ can be turned into a polynomial-time algorithm for arbitrary $|B|$. In particular, the first open case is $|B|=3$. It would also be interesting to determine whether the problem is fixed-parameter tractable when parameterized only by the unary translocation distance. From a practical viewpoint, an important direction is to develop ILP formulations with fewer variables and constraints, which could make exact optimization approaches more effective in practice.

\section*{Declarations}

\textbf{Data availability}. The data and code are available at~\citep{git}, including the instances and scripts used in the computational experiments.

\noindent \textbf{Competing interests}. The authors have no relevant financial or non-financial interests to disclose.

\bibliographystyle{spbasic}

\bibliography{bibl}

\end{document}